\documentclass[letterpaper,11pt,reqno,tbtags]{amsart} 
\usepackage[utf8]{inputenc}
\usepackage[portrait,margin=1in]{geometry} 
\usepackage{mathrsfs,xfrac} 
\usepackage[colorlinks=true,linkcolor=red,citecolor=blue,urlcolor=red]{hyperref} 
\usepackage[mathlines]{lineno} 
\usepackage{amsmath,amssymb,amsthm,amsfonts,amsbsy,latexsym,dsfont,color,booktabs} 
\usepackage[numeric,initials,nobysame]{amsrefs} 
\usepackage{upref,halloweenmath}

\usepackage[utf8]{inputenc}
\usepackage[english]{babel}

\usepackage[textsize=tiny]{todonotes}
\usepackage{regexpatch}
\makeatletter
\xpatchcmd{\@todo}{\setkeys{todonotes}{#1}}{\setkeys{todonotes}{inline,#1}}{}{}
\makeatother

\usepackage{etoolbox}          
\newcommand*\linenomathpatch[1]{%
  \cspreto{#1}{\linenomath}%
  \cspreto{#1*}{\linenomath}%
  \csappto{end#1}{\endlinenomath}%
  \csappto{end#1*}{\endlinenomath}%
}
\newcommand*\linenomathpatchAMS[1]{%
  \cspreto{#1}{\linenomathAMS}%
  \cspreto{#1*}{\linenomathAMS}%
  \csappto{end#1}{\endlinenomath}%
  \csappto{end#1*}{\endlinenomath}%
}
\usepackage{cleveref}

  \crefname{theorem}{Theorem}{Theorems}
  \crefname{thm}{Theorem}{Theorems}
 
  \crefname{lem}{Lemma}{Lemmas}
  \crefname{rem}{Remark}{Remarks}
  \crefname{prop}{Proposition}{Propositions}
  \crefname{proposition}{Proposition}{Propositions}
\crefname{notation}{Notation}{Notations}
\crefname{claim}{Claim}{Claims}
  \crefname{defn}{Definition}{Definitions}
  \crefname{corollary}{Corollary}{Corollaries}
  \crefname{section}{Section}{Sections}
  \crefname{figure}{Figure}{Figures}
  \crefname{question}{Question}{Questions}
  \crefname{exercise}{Exercise}{Exercises}
    \crefname{assumption}{Assumption}{Assumptions}

\newtheorem{thm}{Theorem}[section]

\newtheorem{lem}[thm]{Lemma}
\newtheorem{corollary}[thm]{Corollary}

\newtheorem{defn}[thm]{Definition}

\newtheorem{question}[thm]{Question}

\numberwithin{equation}{section}

\theoremstyle{definition}
\newtheorem{rem}[thm]{Remark}

\expandafter\ifx\linenomath\linenomathWithnumbers
  \let\linenomathAMS\linenomathWithnumbers
  \patchcmd\linenomathAMS{\advance\postdisplaypenalty\linenopenalty}{}{}{}
\else
  \let\linenomathAMS\linenomathNonumbers
\fi

\linenomathpatch{equation}
\linenomathpatchAMS{gather}
\linenomathpatchAMS{multline}
\linenomathpatchAMS{align}
\linenomathpatchAMS{alignat}
\linenomathpatchAMS{flalign}

\makeatletter
\patchcmd{\mmeasure@}{\measuring@true}{
  \measuring@true
  \ifnum-\linenopenaltypar>\interdisplaylinepenalty
    \advance\interdisplaylinepenalty-\linenopenalty
  \fi
  }{}{}
\makeatother

\usepackage{enumerate}

\renewcommand{\leq}{\leqslant} 
\renewcommand{\geq}{\geqslant} 
\renewcommand{\le}{\leqslant} 
\renewcommand{\ge}{\geqslant}

\newcommand{\eps}{\varepsilon}

\newcommand{\norm}[1]{\left\Vert#1\right\Vert}

\let\ga=\alpha \let\gb=\beta  \let\gd=\delta 
     \let\gl=\lambda           \let\gs=\sigma  
  
 \let\gD=\Delta   
\let\gO=\Omega         \let\gS=\Sigma  
                                
\newcommand{\cA}{\mathcal{A}}\newcommand{\cB}{\mathcal{B}}
\newcommand{\cD}{\mathcal{D}}\newcommand{\cE}{\mathcal{E}}\newcommand{\cF}{\mathcal{F}}
\newcommand{\cG}{\mathcal{G}}\newcommand{\cI}{\mathcal{I}}
\newcommand{\cL}{\mathcal{L}}

\newcommand{\cP}{\mathcal{P}}\newcommand{\cR}{\mathcal{R}}

\newcommand{\cV}{\mathcal{V}}
  
\newcommand{\mv}[1]{\boldsymbol{#1}}

\newcommand{\mvH}{\boldsymbol{H}}
\newcommand{\mvJ}{\boldsymbol{J}}
\newcommand{\mvM}{\boldsymbol{M}}

\newcommand{\mvZ}{\boldsymbol{Z}}

\newcommand{\mvy}{\boldsymbol{y}}
\newcommand{\mvz}{\boldsymbol{z}}

\newcommand{\dE}{\mathds{E}}

\newcommand{\dN}{\mathds{N}}
\newcommand{\dP}{\mathds{P}}
\newcommand{\dR}{\mathds{R}}
\newcommand{\dT}{\mathds{T}}

\newcommand{\dZ}{\mathds{Z}} 

\DeclareMathOperator{\E}{\mathds{E}}

\DeclareMathOperator{\var}{Var}
\DeclareMathOperator{\cov}{Cov}

\DeclareMathOperator{\argmin}{argmin}

\def \bf {\boldsymbol}

\title{ Uniqueness and CLT for the  Ground State of the Disordered Monomer-Dimer Model on $\dZ^{d}$}
\author{Kesav Krishnan and Gourab Ray }

\date{\today}
\thanks{ Department of Mathematics and Statistics, University of Victoria, PO Box 1700 STN CSC, Victoria,  V9B 0Z2. Email: gourabray,kkrishnan@uvic.ca.\\
Research supported by NSERC 50711-57400 of GR and PIMS pdf fellowship of KK}
\begin{document}

\begin{abstract}
    We prove that the disordered monomer-dimer model does not admit infinite volume incongruent ground states in $\dZ^d$ which can be obtained as a limit of finite volume ground states. Furthemore, we also prove that these ground states are stable under perturbation of the weights in a precise sense.
    As an application, we obtain a CLT for the ground state weight for a growing sequence of tori. Our motivation stems from a similar and long standing open question for the short range Edwards-Anderson spin glass model.
\end{abstract}

\maketitle

\section{Introduction}
Understanding ground states for disordered statistical physics models has been one of the most tantalizing research topics in the past few decades. A notable example is the \emph{Edwards Anderson spin glass model}, which is essentially the well-studied Ising model but with the couplings given by i.i.d.\ Gaussian weights. Despite fantastic progress of rigoruous mathematical understanding of the mean field case \cites{sherrington_kirkpatrick,parisi,talagrand,panchenko}, progress has been slow in the short range setting (i.e. on $\dZ^d$), and in fact has even been a source of considerable controversy in the physics literature (see \cite{chatterjee2023spin} and references therein). The goal of this article is to address the question of uniqueness and central limit theorem for the ground state of the  disordered \textbf{monomer-dimer model} in the short range setup. We also shed light on the question of \emph{disorder chaos} for this particular model.

Let $G = (V,E)$ be a finite, connected, simple graph and let $\Sigma = V \cup E$. A monomer-dimer covering of $G$ is defined by a collection of non overlapping edges and vertices of $G$, which we denote by $M\subsetneq \Sigma$, otherwise known as a \textbf{matching}. By non overlapping, we mean that no two distinct edges in $M$ share an endpoint and no vertex in $M$ is an endpoint of an edge in $M$. The vertices in $M$ are called \textbf{monomers} and the edges \textbf{dimers}. 

\begin{figure}[t]
    \centering
    \includegraphics[scale=0.4]{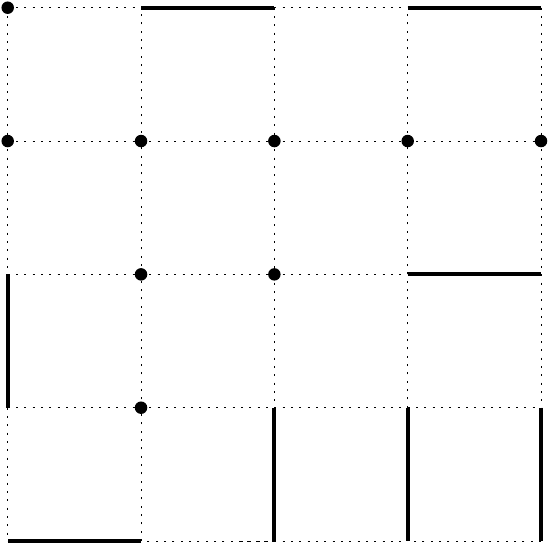}
    \caption{A monomer-dimer configuration on the $5\times 5$ grid}
    \label{fig:md}
\end{figure}
We will also think of $M$ as a percolation on $\Sigma$, which allows us to think of it as an element of $\{0,1\}^\Sigma$, where a 1 corresponds to an edge or a vertex being occupied and $0 $ corresponds to them being unoccupoied. Given a collection of weights $J:\Sigma \to \dR$ on a finite graph the weight of a matching is given by 
\[
H(M):=\sum_{x\in M}J_{x}. 
\]
The ground state of the monomer-dimer model on $G$ with weights $J$ is any matching of minimal weight. On finite graphs, if the weights are independent random variables with continuous distribution, there is zero probability that any two distinct monomer dimer configurations will have the same weight, thus uniqueness of the ground state holds with probability $1$. 

On infinite graphs, while the notion of weight of the matching does not make sense any more, following \cite{AW_90,NS_01}, we can still extend notion of ground state using the energy difference of two configurations that differ at finitely many sites.  Namely, we can define a matching $M\subset V \cup E $ to be a ground state if we cannot `lower the energy' by modifying it at finitely many edges and vertices, where the lowering of the energy is given precise meaning by considering energy differences (see \Cref{sec:infinite_ground}). If the weights come from i.i.d.\ distribution then the cardinality of the set of ground states ${\sf G}(\mvJ)$ is almost surely a constant via ergodicity. What can we say about this number, in particular, is it finite or infinite? We answer this question for the particular case of ground states which are obtained through an infinite volume limit of finite volume tori. 

In this paper, we work on the integer lattice $\dZ^d$, $d \ge 2$ with nearest neighbor edges. Let $\sf V$ be the vertex set and $\sf E$ be the edge set of $\dZ^d$, with $\sf \Sigma  = \sf V \cup \sf E$. We let $\dT_n^d$ to be the torus $(\dZ/n\dZ)^d$. For ease of proof technique, we let ${\bf J}= (J_x)_{x \in \sf \Sigma}$ be a collection of i.i.d.\ weights coming from a \textbf{good} distribution, see \Cref{def:good}. For now let us mention that this includes a large class of distributions, in particular the cases of $N(0,1)$ (standard Gaussian) and Exp$(1)$ (standard exponential).   
Let $M_n$ denote the ground state for the monomer-dimer configuration on for the weights restricted to the vertices and edges of $\dT_n^d$. We think of $(\bf J, M_n)$ as a joint probability distribution on $[\gb,\infty)^{\sf \Sigma} \times \{0,1\}^{\sf \Sigma}$ where $M_n$ is extended periodically to the infinite lattice. It is standard that the collection of probability measures induced by $\{(\bf J, M_n)\}_{n \ge 1}$ is tight and hence has subsequential limits. Our main theorem concerns the uniqueness of these limits.
\begin{thm}\label{thm:main}
    The weak limit of $\{(\bf J, M_n)\}_{n \ge 1}$ exists as $n \to \infty$. Furthermore, letting $(\bf J, M)$ be the limit, $M$ is a measurable function of $\bf J$.
\end{thm}
Even though each $M_n$ is clearly a measurable function of $\bf J$, measurability might not hold for weak limits in general. For example, on $\dZ$, we can define a matching on $\dT_{2n}$ as follows: the matching is a perfect matching (only dimers and no monomers) if the sum of the weights on $\dT_{2n}$ is strictly positive, and otherwise, the matching is all monomers and no dimers. It is easy to see that in the limit, we get all monomers with probability $1/2$ and a perfect matching  with probability $1/2$, which is in fact completely independent of ${\bf J}$.

The broad strategy to prove \Cref{thm:main} is as in \cite{NS_01}: superimpose two potential ground states and arrive at a contradiction.  To that end let $({\bf J}, M)$ and $({\bf J}, M')$ be two distinct subsequential limits. Conditioned on ${\bf J}$, sample $M$ and $M'$ independently according to their conditional laws (the spaces are Borel which allows us to consider the a.s. conditional laws). This defines a coupling $({\bf J}, M, M')$. We prove the following theorem
\begin{thm}\label{thm:conditional_equality}
    Let  $({\bf J}, M, M')$ be as above. Then $M= M'$ almost surely.
\end{thm}
It is easy to see that \Cref{thm:main} follows immediately from \Cref{thm:conditional_equality}. Indeed, since $M$ and $M'$ can be coupled to be equal, all subsequential limits must be the same. Furthermore, taking $M$ and $M'$ to be limits along the same subsequence, we conclude that $M$ is a measurable function of ${\bf J}$.

Our next result concerns the change in the ground state under a translation invariant perturbation. To that end, for any two elements $\xi , \xi' \in \{0,1\}^{\sf \Sigma}$, let $\xi \Delta \xi' := \{x \in \Sigma: \xi_x \neq \xi'_x\}$ denote their symmetric difference. Let ${\bf J'}$ be an independent copy of ${\bf J}$. Fix $p\in [0,1]$ and now define ${\bf J}(p)$ as follows: for each edge toss an independent coin and change $J_e$ to $J_e'$ with probability $p$. Note that ${\bf J}(p)$ has the same law as ${\bf J}$. We get a new ground state $M(p)$ using \Cref{thm:main}.

\begin{thm}\label{thm:perturbation}
Fix $p \in [0,1]$ and let $M,M(p)$ be as above. Then $M \Delta M(p)$ has finite components almost surely.  Furthermore, $M\Delta M(p)\to \emptyset$ as $p \to 0$ in law.
\end{thm}
The result $M\Delta M(p)\to \emptyset$ as $p \to 0$ in law is analogous to the absence of disorder chaos (see \cite{BM87,chatterjee2023spin} and references therein). 


As an application of our results, we obtain a central limit theorem for the  partition function.

\begin{thm}\label{thm:CLT}
Let $M_{n}$ be the ground state for the monomer dimer model with $\mvJ$ now additionally satisfying the condition that $\E(|J_{x}|^{4+\gd})<\infty$ for some $\gd>0$. Then as $n\to \infty$,
\[
\frac{H(M_{n})-\E H(M_{n})}{\sqrt{\var{H(M_{n})}}}\to N(0,1)
\]
in distribution. 
\end{thm}

As a by product of our proof, we also obtain a decay of correlation for the ground state, see \Cref{cor:decay_corr} for a precise statemtent.
\begin{rem}
For each of our results, it is possible to consider the weight distributions to be different for the vertices and edges. As we will touch upon later, it is also possible to choose the monomer weights to be a fixed constant. See \Cref{sec:generalization} for details.
\end{rem}

The monomer dimer model in the disordered setting has been studied before in the context of a \emph{Gibbs measure} \cite{ac1,ac2,DK_23,SL_23}. The mean field setting with two distinct types of environmental disorder, the first being studying the standard monomer dimer model on the Erdos-Renyi random graph, and the second corresponding to the vertex weights being i.i.d random variables are addressed in \cite{ac1} and \cite{ac2} respectively. In each case, the limiting free energy and limiting monomer densities are computed explicitly. The approaches used are inherently mean field, in particular using the cavity method to express the limiting free energy in terms of a fixed point relation. In \cite{DK_23,SL_23}, the authors establish the finite temperature analogue of \Cref{thm:CLT} in the setting of pseudo one dimensional graphs and general bounded degree graphs respectively, and in \cite{DK_23}, the CLT for the ground state energy is briefly addressed. However, none of these works analyze the ground state in detail, which is far more rigid in structure.

Our proof is more inspired by the technique of Newman and Stein \cite{NS_01} where they analyze the superimposition of two ground states in dimension 2 and derive a contradiction via ergodic theoretic methods. Although their proof does not yield the uniqueness of the ground state for the Edwards Anderson model, which is a major open question in the area of spin glasses with far reaching consequences, the special structure of the monomer dimer model does allow us to obtain uniqueness in any dimension. In contrast, it is unclear whether the ground state  for the Edwards Anderson model is in fact unique in any dimension (see \cite{chatterjee2023spin} for a recent account), although it is strongly believed to be the case in dimension 2. Our proof also uses the concept of \emph{metastates} put forth by Aizenmann and Wehr \cite{AW_90}, which is a crucial concept used to capture the chaotic behaviour of the ground states. From a broader perspective, the ground state of the disordered monomer-dimer model is the solution of a randomized combinatorial optimization problem, in the spirit of \cites{chatterjee_sen_mst,cao,aldous_zeta_2,dey2024random}.

\subsection{Outline of the proof}\label{sec:outline}
One convenient aspect of using tori to obtain limits is that $({\bf J},M,M')$ is  invariant in law with respect to translations of the lattice, which enables us to use ergodic theoretic methods. Following Newman and Stein \cite{NS_01} we consider $M \Delta M'$. It follows immediately from the fact that $M$ and $M'$ are ground states that $M \Delta M'$ cannot have any finite component almost surely. Furthermore, any infinite component of $M \Delta M'$ can be either a (one-ended) infinite path of dimers started from a monomer, or a bi-infinite simple path of dimers.
We first use the celebrated technique of Burton and Keane \cite{BK} to rule out the existence of the one-ended paths (in fact this is a much more straightforward consequence of the amenability of $\dZ^d$, we do not need the full force of the Burton-Keane technique). This in turn rules out the presence of any monomer in $M \Delta M'$.

It remains to rule out bi-infinite path. Here we crucially use the fact that we are in the monomer-dimer setup and define a quantity which we call the \textbf{optimality} of a vertex, denoted $O(v)$ (see \Cref{def:o}). Optimality is  a local function of the weights and has the property that $O(v)<0$ implies that $v$ has to be a monomer in both $M$ and $M'$. Consequently, none of the monomers in any bi-infinite path of $M\Delta M'$ can have negative optimality.

Another concept we borrow from Newman and Stein is the notion of an \textbf{inaccessible} edge, which is the analogue of a super satisfied edge from \cite{NS_01}. Namely an edge is inaccessible if its weight is so high compared to the two monomers adjacent to it that no ground state can occupy it. 

The next notion we introduce the idea of \textbf{flexibility} of a vertex or an edge, which tells us the amount by which we need to change the weight of that vertex or edge in order to change the ground state (the analogue of this idea for Edwards Anderson model was introduced in \cite{NS_01,ADNS_10}). One needs to be careful with this definition in infinite volume limits which is done by taking appropriate weak limits along subsequences borrowing from the notion of metastates going back to \cite{AW_90}. A crucial property of flexibility is that if we change the weight of some $x \in \sf \Sigma$ by $\eps$, the paths obtained by superimposing the new and the old ground states cannot pass through an edge of flexibility strictly bigger than $\eps$.

Let us restrict to the case where the density is positive on all of $\dR$ for simplicity, the other cases need a slight tweak which we don't mention here.
The main argument is now as follows; we restrict to the event that a bi-infinite path $P$ passes through the origin, and fix $\eps>0$ so small so that there are infinitely many edges of flexibilities at least $\eps$ along both directions of $P$. Find the two nearest ones to the origin and let $Q$ be the portion of $P$ between these edges of high flexibility. Next we let $m$ be the infimum of the optimality of the vertices along $P$, which is non-negative as argued above. We further restrict to the event that the origin has optimality at most ${\sf m}+\delta$ with $\delta<\eps/2$ and all the edges adjacent to vertices in $Q$ except the origin are inaccessible. We show that this event has positive probability. Now we decrease the weight of the origin by $2\delta$ to obtain a new collection of weights $\tilde {\bf J}$ and new ground states $\tilde M, \tilde M'$. We prove that $\tilde {\bf J}$ is absolutely continuous with respect to ${\bf J}$ and with some more work $(\tilde {\bf J}, \tilde M, \tilde M')$ is absolutely continuous with respect to $({\bf J}, M,M')$ as well.

The perturbation causes one of two things to happen to the ground state; either it changes, or it doesn't. In the former case, due to topological constraints, a monomer has to appear in $\tilde M \Delta \tilde M'$. This is a contradiction as $M \Delta M'$ cannot have a monomer with positive probability and the law of $\tilde M \gD \tilde M'$ is absolutely continuous with that of $M\gD M'$. The crucial thing we need to rule out for this step is that lowering the weight of the origin does not wipe out the entire bi-infinite path through it.

In the latter case when there is no change to the ground state, we are on an event where the infimum of optimality over all the vertices in the bi-infinite path of $\tilde M \Delta \tilde M'$ passing through the origin is achieved at some vertex, an event which has probability 0 in $M\Delta M'$ because of translation invariance. The contradiction again follows from absolute continuity.

The proof of \Cref{thm:perturbation} follows similar lines, except we are allowed to have finite components in this setting. The proof of \Cref{thm:CLT} uses the powerful Normal approximation technique put forth by Chatterjee \cite{C_Normal_08}, and the crucial input of \Cref{thm:main} which yields a mild form of spatial decorrelation of the ground states.

\subsection{Outline of the paper:} In \Cref{sec:finite} we derive some basic lemmas about the ground state in finite graphs. In \Cref{sec:infinite_ground}, we extend these properties to infinite graph ground states. In \Cref{sec:uniqueness} we prove \Cref{thm:conditional_equality}. In \Cref{sec:perturbation}, we prove \Cref{thm:perturbation}. In \Cref{sec:CLT} we prove \Cref{thm:CLT}. We finish with some discussion about generalizations in \Cref{sec:generalization} and some open questions in \Cref{sec:open}.
\subsection*{Notations:} We denote by $\mathsf V, \mathsf E$ respectively the vertex and edge sets of $\dZ^d$ and by $\mathsf V_n, \mathsf E_n$, the vertex and edge sets of the torus $(\dZ/n\dZ)^d$. Let ${\mathsf \Sigma}  = \mathsf V \cup \mathsf E$ and $\Sigma_n = \mathsf V_n \cup \mathsf E_n$ Let $\Omega_G$ be the set of matchings of a graph $G$.  Let $B_K$ denote the box $[-k,k]^d \cap \dZ^d$.

\section{Perturbation of the Ground State on Finite Graphs}\label{sec:finite}
In this section, we will provide certain basic lemmas that primarily concern the change of the ground state when one or more of the weights are perturbed. The difference between two ground states is clearly encoded in the symmetric difference between them.  {Throughout this section, we assume that $G = (V,E)$ is a finite, connected, simple graph,\ with $\gS = V \cup E$. Sometimes we shall write $w_{u,v}$ to be the weight of $(u,v)$ admitting an abuse of notation}. We also denote by $\Omega_G$ the set of all matchings of $G$. A collection of weights $\{w_{x}\}_{x\in \gS}$ is said to be generic if for every finite set $S \subset \Sigma$ and every collection of integers $\{k_{x}\}_{x\in S}$ not all 0, 
\[
\sum_{x\in S} k_{x}w_{x}\neq 0.
\]
Throughout this section, we assume that the weight collection $W=\{w_x\}_{x \in \gS}$ are deterministic and \textbf{generic}. 
Clearly, the weights being generic guarantees the uniqueness of the ground state, and any collection of i.i.d. continuous random variables are generic almost surely. 
A \textbf{simple, finite path} is a sequence   $(u_1,(u_1,u_2),\ldots, (u_{k-1},u_k), u_k)$ in $\Sigma$ so that and $u_i$s are all distinct vertices. A \textbf{simple loop} is a finite sequence of edges $((u_1,u_2), (u_2,u_3), \ldots (u_{k-1}, u_1))$ so that $u_i$s are all distinct.

Given a finite, simple path $P:=(u_1,(u_1,u_2),\ldots, (u_{k-1},u_k), u_k)$, we can 
partition it into two subsets $P_1:=\{u_1, (u_2,u_3), (u_4,u_5), \ldots, (u_{k-2}, u_{k-1}), u_k\}$ and $P_2:=\{(u_1,u_2),(u_3,u_4), \ldots, (u_{k-1}, u_k)\}$. Call the pair $P_1,P_2$, \emph{complementary matchings of $P$.} Similarly, given a simple finite loop $L := (u_1,u_2, \ldots, u_k= u_1)$, we can decompose its edges into $L_1:=(u_1,u_2), (u_3,u_4), \ldots $ and the remaining $L_2:= (u_2,u_3), (u_4,u_5), \ldots$, and call the pair $(L_1,L_2)$ complementary matchings of $L$. Observe that if one of $L_i$ ( or $P_i$) is in a matching, we can remove every element of $L_i$ (or $P_i$) from it, and add back its complementary pair, to get another valid matching. 
Let $N(v)$ denote the \emph{neighborhood} of $v$ which includes $v$ and all the edges adjacent to $v$. We start  with a few elementary lemmas.
\begin{figure}
    \centering
    \includegraphics[scale=0.7]{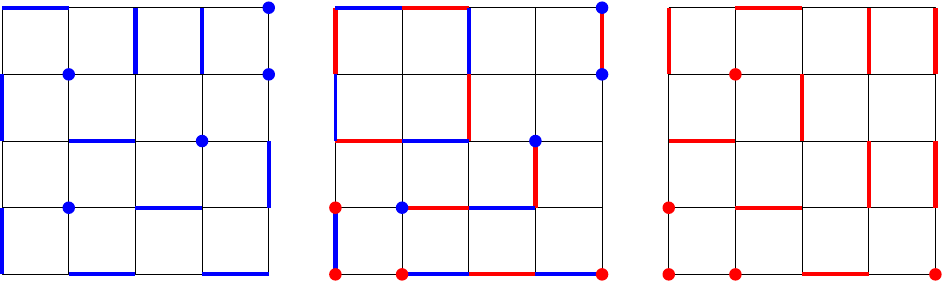}
    \caption{$M_{1}$ is depicted on the left in blue, $M_{2}$ is depicted on the right in red. The symmetric difference $M_{1}\gD M_{2}$ is depicted in the middle.}
    \label{fig:enter-label} 
\end{figure}
\begin{lem}\label{lem:pathsloops}
The symmetric difference of any two matchings $M_{1}$ and $M_{2}$ is a disjoint collection of  simple, finite paths and  simple loops.

Furthermore, in each component, the set of edges and vertices in $M_1$ and those in $M_2$ form a complementary pair.
\end{lem}
\begin{proof}
    This is elementary as for any vertex $v$, if $N(v) \cap (M_1\Delta M_2 ) \neq \emptyset$, then exactly one of the following two cases occur. Either $v$ is occupied by one matching and exactly one edge adjacent to $v $ is occupied by the other; or $v$ is occupied by neither matchings and exactly one edge adjacent to $v$ is occupied by $M_1$ and another different edge is occupied by $M_2$. Thus the set of edges of $M_1 \Delta M_2$ is a subgraph with degree either $0,1$ or $2$. Furthermore if a vertex $v$ has degree 0, both matchings occupy $v$, if it has degree 1, exactly one of them occupy $v$, and if it has degree 2, none of them occupy $v$. It is easy to see that these constraints make every component of $M_1\Delta M_2$ either a simple, finite path or a simple loop. 

    The second assertion is a straightforward from the definition of a complementary pair.
\end{proof}

In light of \Cref{lem:pathsloops}, we call each element  of the disjoint collection of finite, simple paths and loops a \textbf{component} of the symmetric difference between two matchings. 

\begin{lem}\label{lem:obvious}
Let $M$ be the ground state and assume $v \in M$. Suppose we change the weight at $v$ to $\tilde{w}_{v}$, {so that it remains generic and the new ground state is $\tilde{M}$.} Then the symmetric difference $M \gD \tilde{M}$ is either empty or is {a simple, finite path} beginning at the vertex $v$. 
\end{lem}
\begin{proof}
Using \Cref{lem:pathsloops}, another way to state the lemma is to claim that there is no component of $M \Delta \tilde{M}$ which do not contain $v$, which is what we shall prove. Take one such component $C$. Since the weight of all the edges and vertices in $C$ is the same for $M$ and $\tilde{M}$, switching between the complementary matchings of $C$ either decreases the weight of $M$ or $M'$, contradicting the minimality of their weights. 
\end{proof}

Prior to completely delving into the issue of perturbations that cause the ground state configuration to change, we state the following notion of \textbf{gauge invariance} for the monomer-dimer model. 
\begin{lem}\label{lem:gauge}
Let $M_{G}$ denote the ground state associated to a system of weights $\{w_{x}\}_{x\in \gS}$. Let $u\in V$ be a vertex, and let $E_{u}$ denote the collection of edges incident to $u$. For a fixed $\gl\in \dR$, consider the new system of weights given by 
\[
w^{u,\gl}_{x}:=\begin{cases}w_{x}+\gl \text{ if $x\in\{u\}\cup E_{u}$ }\\ w_{x} \text{ otherwise}. \end{cases}
\]
The ground state $M_{G}$ remains the same for the new collection of weights. 
\end{lem}
\begin{proof}
Note that every matching $M\in \gO_{G}$ must either contain $u$ itself or exactly one edge contained in $E_{u}$. Thus, $\forall M\in \gO_{G}$, $H^{u,\gl}(M)=H(M)+\gl$, where $H^{u,\gl}$ is the energy functional defined with respect to the new weights. Thus, although the energies have changed, the energy of every matching has been shifted by $\gl$, and the energy differences remain the same. $M_{G}$ remains to be the ground state. 
\end{proof}
The question of stability of certain vertices and edges with regards to modification of the weights is natural {and central to our analysis,} and we will explore this now. We write $u \sim v$ if $u$ is adjacent to $v$.
\begin{defn}\label{def:o}
    The \textbf{optimality} of $v$, denoted $O(v)$ is defined as \begin{equation}\label{eq:opt}
O(v):=\max_{u:u\sim v} \left(w_{u}+w_{v}-w_{u,v}\right).
\end{equation}
An edge $e  = (u,v)$ is called \textbf{inaccessible} if 
\begin{align}\label{eq:inaccesible}
w_e > w_u + w_v.
\end{align}
\end{defn}

Note that both the  notions in \Cref{def:o}  are completely local which will be useful later in the infinite volume setup. The next lemma justifies the nomenclature.

\begin{lem} Let $v\in V, e \in E$. If $e = (u,v)$ is inaccessible, then $e$ is not in the ground state matching.
If $O(v)<0$, then $v$ must be in the ground state matching. 
\end{lem}
 
\begin{proof}
If $e = (u,v)$ is inaccessible and is in a  matching, then we can lower the energy by removing it and adding back two monomers at $u$ and $v$. Thus $e$ cannot be in a ground state. If $O(v)<0$ then all the edges adjacent to $v$ are inaccessible, and we apply the previous result.
\end{proof}

\subsection{Excited States}\label{sec:excitedstates}
For any  $S \subset \gS$ we define the external boundary of $S$, denoted $\partial_{\textrm{ext}}S$ as the collection 
\[
\partial_{\textrm{ext}}S:=\{x:x\in S^{c}, x\sim y\in S\}. 
\] 
We will define for any set $S$,
\[
\overline{S}:=S\cup \partial_{\textrm{ext}}S.
\]
Next, we consider configurations $\xi:S\to \{0,1\}$, which are \textbf{valid} in the sense that they are extendable to yield a matching. That is, $\xi$ is said to be valid if $\exists \zeta \in \gO_{G}$ such that $\zeta(x)=\xi(x)$ for all $x\in S$. (For example, we are not allowed to have a vertex  $v$ and all its neighborhood to be in $S$ and empty.) Let $\cV(S)$ be the set of all valid configurations in $\{0,1\}^S$. We define $M_{G,S, \xi}$ to be the ground state with the constraint $M_{G,S,\xi} (x) = \xi(x)$ for all $x \in S$. More precisely,  
\begin{align}
M_{G,S,\xi}:=\argmin\{H(\zeta): \zeta \in \gO_{G},\text{ }\zeta(x)=\xi(x)\text{ }\forall x\in S\}. 
\end{align}
Note that since the weights are generic, the constrained ground state is still unique. We also define 
$$
\Delta H(G,S,\xi,\xi') = H(M_{G,S, \xi'})-H(M_{G,S, \xi})
$$
to be the energy difference between two ground states with constraints $\xi$ and $\xi'$ on $S$. For a given $S$, the energy difference enables us to examine how the ground state changes as $\{w_{x}\}_{x\in S}$ are varied, with all other weights remaining fixed . Consider the following two simple observations:
\begin{itemize}
\item For local configurations $\xi,\xi'$ and $\xi'' \in \{0,1\}^{S}$, 
\begin{equation}
    \Delta H(G,S,\xi,\xi') +\Delta H(G,S,\xi',\xi'') = \Delta H(G,S,\xi,\xi'').\label{eq:E_linear}
\end{equation}
\item $M_{G,S,\xi}$ does not depend on the weights $(w(x))_{x \in S}$. 
\end{itemize}
Using these elementary properties, we can partition $\dR^S$ into sets where $M_G=M_{G,S,\xi}$ as follows. Define
$$
\cI(G,S,\xi) = \{w_S: \Delta H(G,S, \xi, \xi') > 0 \text{ for all }\xi' \in \cV(S), \xi' \neq \xi \}.
$$
Let $B(n,S) = \{w_S \in \dR^S: \Delta H(G,S, \xi, \xi') = 0 \text{ for some }\xi \neq \xi', \xi, \xi' \in \cV(S)\}$.
The following lemma is clear now since the weights are chosen to be generic.
\begin{lem}\label{lem:partition}
    $\{\cI(n,S,\xi)\}_{\xi \in \cV(S)} \cup B(n,S)$ is a partition of $\dR^S$ and furthermore $w_S \not \in B(n,S)$. 
\end{lem}
The following equality is clear from the definitions:
\begin{equation}
    M_G = \sum_{\xi \in \cV(S)} M_{G,S,\xi}1_{w_S \in \cI(G,S,\xi)}.\label{eq:metastate}
\end{equation}

The following elementary corollary will be useful later when we discuss absolute continuity of limiting measures on matchings.
\begin{corollary}\label{cor:modification}
    Let $\{w_{x}\}_{x\in \gS}$ be a collection of generic weights and $S\subset \gS$ be a finite subset. Let $f:\dR^{S}\to \dR^{S}$ be a measurable function such that the new weights defined by $\tilde{w}_{x}=f(w_{x})$ for all $x\in S$, all other weights unchanged, remain generic. Then the ground state with respect to the new weights is given by $M_{G,S,\xi}$ where $\xi$ is such that 
    \[f(w_{S})\in \cI(G,S,\xi). 
\]
\end{corollary}
Another extremely convenient feature, when it comes to taking the infinite volume limits is the tightness of $\gD H(G,S,\xi,\xi')$ in the size of $G$ (which will be applied later to a growing sequence of tori). To that end, we first need the following combinatorial lemma, which essentially says that two different valid configurations in $S$ can be extended to be the same outside a small neighborhood of $S$.

\begin{lem} \label{lem:local_extension}
    Let $\xi, \xi' \in \cV(S)$ which are extendable to $G$. Given any extension  $m_\xi$ of $\xi$, there exists an extension $m_{\xi,\xi'}$ of $\xi'$ such that $m_{\xi} \Delta m_{\xi, \xi'} \subset \partial_{\textrm{ext}} \bar S$. 
\end{lem}
\begin{proof}
    Define $m_{\xi, \xi'}$ as follows. We begin with $m_{\xi}$ and replace the configuration for all $x\in S$ with $\xi'$. If at the end of this replacement we are left with a valid matching, we are done. If not, firstly, consider all vertices $v\in V\cap S$ such that $\xi'(v)=0$ and $\xi'(e)=0$ for all $e\in E\cap S$ where $v$ is an end point of $e$.  For every such $v\in S$, we add an edge $(v,w)$ for some $w\notin S$, and $w$ uncovered by any dimer adjacent to $S$. Such a $w$ will always exist, because otherwise $\xi'$ would not be extendable. The introduction of these edges, as well as the prior replacement can result in adjacent dimer pairs or an adjacent monomer dimer pair. Such defects may at most extend to $\partial_{\textrm{ext}}\overline{S}$, and we deal with them now. For every covered edge $e$ connecting a vertex in $\partial_{\textrm{ext}}S$ and $\partial_{\textrm{ext}}\overline{S}$, we perform a regularizing procedure. If both endpoints of $e$ are covered by other monomers/dimers, we remove $e$ from the configuration. If only one endpoint is otherwise covered, we remove $e$ and place a monomer at the now uncovered remaining vertex. Next if $v\in \partial_{\textrm{ext}}{S}\cup\partial_{\textrm{ext}}\overline{S}$, is a vertex that is covered by a monomer as well as an adjacent dimer, we remove the monomer from the configuration.
\end{proof}

Let $M_{G,S,\xi,\xi'}$ be a matching which is defined as in \Cref{lem:local_extension} for $m_\xi = M_{G, \xi}$.  The following trivial bound now holds:
\begin{equation}
   | \Delta H(G,S,\xi,\xi')| \le  \max\{|H(M_{G,S,\xi,\xi'}) - H(M_{G,S,\xi})|, |H(M_{G,S,\xi',\xi}) - H(M_{G,S,\xi'})|\} \le \sum_{x \in \overline{\overline{S} }}|w_e| .\label{eq:energy_diff_tight}
\end{equation}

Let us now focus on the case where $S$ is a single $x\in \gS$. In this case, the partition in \Cref{lem:partition} is simply two intervals of the form $(-\infty, K_{G,x}) \cup (K_{G,x},\infty) \cup \{K_{G,x}\}$ where $K_{G,x}$ is called the \textbf{transition point} at $x$. Using this notation, \eqref{eq:metastate} takes the following simple form:
\begin{equation}
    M_G = M_{G,x,0}1_{w_x > K_{G,x}} + M_{G,x,1}1_{w_x \le  K_{G,x}}\label{eq:metastate_e}
\end{equation}

{
\begin{defn}\label{def:Mx}
For any $x \in \Sigma$, we define \textbf{flexibility} of $x$ to be $F_G(x):=|K_{G,x} - w_x|=|\gD H(G,x,0,1)|$. 
\end{defn}

Both nomenclature of the transition point and the flexibility follows the usage of Newman and Stein \cite{NS_01}. In fact, it is not too hard to see that we can write down the following formula for $K_{G,x}$:
\begin{equation}
    K_{G,x} = H(M_{G,x,0}) - H(M_{G,x,1})+w_x.\label{eq:flexibility}
\end{equation}

We note here several elementary properties of $F_G(x)$, whose proof is straightforward application of the definition, so we skip it.
\begin{lem}\label{lem:flex_elementary}
Fix $x \in \Sigma$.
\begin{enumerate}[i.]
    \item $K_{G,x}$ if a function of $\{w_{y}: y \in \Sigma \setminus \{x\}\}$. 
    \item $x$ is in the ground state if and only if $w_x \le K_{G,x}$.
\end{enumerate}
\end{lem}
Now let us argue that perturbing the weight of an edge or a vertex by a quantity smaller than the flexibility of $x\in \gS$ will not change the status of $x$. This lemma justifies the use of the nomenclature `flexibility'.
\begin{lem}\label{lem:perturb}
    Suppose $M$ is the ground state and $x \in \Sigma$. Suppose the weight of some $y \in \Sigma$ is either increased or decreased by $\eps < F(x)$, to get a possibly new ground state $\tilde M$. Then $x \notin M\gD \tilde{M}$.
\end{lem}
\begin{rem}
Note the lemma when applied to $x=y$ simply yields $M\gD\tilde{M}=\emptyset$.
\end{rem}
\begin{proof}
    We observe from  \eqref{eq:flexibility} that changing the weight of a vertex or an edge by $\eps$ changes $K_{G,x}$ by at most $\eps$. Indeed, if $y=x$, $K_{G,x}$ is unchanged. On the other hand, if $y\neq x$ then if $y$ is present in both $M_{G,x,1} $ or $M_{G,x,0}$, or if it is present in neither, $K_{G,x}$ does not change. However, if it is present in one of $M_{G,x,1} $ or $M_{G,x,0}$, $K_{G,x}$ changes by at most $\eps$ and $w_x$ does not change.  Now we can use \Cref{lem:flex_elementary}, item $ii.$ and conclude.
\end{proof}
}

{The next lemma is fairly intuitive, we include it for the sake of exposition. It concerns the robustness of the flexibility with regards to the raising of the weights of edges (or vertices) not selected by the ground state.  
\begin{lem}\label{lem:flex2}
Let $M$ be the ground state, let $S \subset \Sigma $ such that no element of $S$ is  in $M$.  Consider new weights $w'_{y}\ge w_y$ for all $y \in S$. Then the ground state is unchanged and furthermore for any $x\in \Sigma$, $F_G(x)$ for the new weights cannot be strictly smaller than that for the old weights. 
\end{lem}
\begin{proof}
  If $x\in M$, then $ M = M_{G,x,1}$ and by \Cref{lem:flex_elementary}, item $ii.$ $F_G(x) = K_{G,x} -w_x$. Using the formula \eqref{eq:flexibility} of $K_{G,x}$, clearly  $K_{G,x}$ cannot decrease by the given change of weights and hence $w_x $ remains smaller than $K_{G,x}$ for the new weights. Thus the ground state remain unchanged.

  Note that \eqref{eq:flexibility} also entails that 
  $F_G(x) = H(M_{G,x,0})  - H(M_{G,x,1})$ if $x \in M$ and $F_G(x) = H(M_{G,x,1})  - H(M_{G,x,0})$ otherwise.
  It is easy to see from this formula that $F_{G}(x)$ cannot decrease.
\end{proof}
}
We conclude this section with a lemma which perhaps seems contrived right now, but will be a crucial ingredient of the proof of Theorem \ref{thm:main}. 

\begin{lem}\label{lem:pathmod}

Let $S :=(u_1,(u_1,u_2),\ldots, (u_{k-1},u_k), u_k)$ be a simple, finite path and let $1<j< k$ be an integer such that the following conditions hold
\begin{itemize}
\item The edges of $S$ are alternatingly occupied by dimers corresponding to the ground state $M$.
\item $(u_j,u_{j+1})$ is occupied by a dimer. 
\item   $F_G((u_{k-1},u_k))\geq \eps$ for some $\eps$. 
\end{itemize}
Let $E_S$ be the set of edges adjacent to $(u_2,u_3, \ldots, u_{j-1}, u_{j+1},\ldots, u_{k-1})$ which are not in $S$. Let $\{\tilde w_{x}\}_{x\in \gS}$ denotes a new set of generic weights defined as follows:
\begin{itemize}
    
    \item  $\tilde w_e  = w_e+z_{P}$  for all $e \in E_S$, $z_{S}>0$ chosen to make every edge in $E_S$ inaccessible
    \item $\tilde{w}_{(u_1,v)}={w}_{(u_1,v)}+\eps/2\text{ }\forall v\sim u_{1}$ 
    \item All other weights are unchanged.
\end{itemize}  
Let $\tilde M$ be the ground state with respect to the new weights. Then $\tilde{O}(u_1)=O(u_1)-\eps/2$, and the following dichotomy holds for $M\gD \tilde{M}$. Either
\begin{itemize}
\item $M\gD \tilde{M}=\emptyset$, or,
\item $M\gD \tilde{M}= (u_j,(u_j,u_{j+1}), \ldots, (u_{p-1}, u_p), u_p)$ for some $p \le k-1$.
\end{itemize}

\end{lem}
\begin{proof}
The order of changing the weights in this context does not matter as the ground state is unique. We first apply a gauge transformation at $u_{1}$, and add $\eps/2$ to the weights at ${u_{1}}$ as well as all adjacent edges. This of course preserves the ground state by Lemma \ref{lem:gauge}. We next change the weights for all the edges in $E_{S}$. We choose
\begin{align}\label{eq:zs}
z_{S}:=\max_{(u_{i},v)\in E_{S}}|w^{u_{1},\eps/2}_{(u_{i},v)}-w^{u_{1},\eps/2}_{u_{i}}-w^{u_{1},\eps/2}_{v}|+V,
\end{align}
where $V$ is a uniform $[0,1]$ random variable, and define $\tilde{w}_{e}$ for this choice of $z_{S}$. $V$ is included so that the system of weights $\{\tilde{w}_{x}\}$ remains generic. It is a consequence of Lemma \ref{lem:flex2} that the ground state remains unchanged, and it is immediately clear that all $e\in E_{S}$ are inaccessible as per \eqref{eq:inaccesible}. Next, we decrease the weight of $u_{1}$ by $\eps/2$. The fact that $\tilde{O}(u_1)=O(u_1)-\eps/2$ is trivial and follows directly from the definition of optimality \eqref{eq:opt}. Since  we are changing the weight of the vertex $u_1$, we know that if $M\gD \tilde{M} \neq \emptyset$, $u_1\in M\gD \tilde{M}$, and furthermore $M\gD \tilde{M}$ must be a simple, finite path beginning at $u_1$ by \Cref{lem:obvious}. We denote this path by $Q$. Since $(u_1,u_2)\in M$, we must have that $u_1\in \tilde{M}$ and thus the path $Q$ begins at $u_1$ and contains the edge $(u_1,u_2)$. Now, none of the other edges that are adjacent to $u_2$ can be in $M\gD M'$, since the adjacent weights are unaffected and they remain inaccessible. Therefore, the path $Q$ must proceed along the path $S$. Finally, by Lemma \ref{lem:perturb}, $(u_{k-1}, u_k)\notin Q$ since $F_G((u_{k-1}, u_k))\geq \eps$. Thus, $Q$ must terminate before reaching $u_{k-1}$. This completes the proof. 
\end{proof}
\section{Ground State on Infinite Graphs and Modifications}
\subsection{Defining the ground state}\label{sec:infinite_ground}
In this section go back to the Euclidean lattice $\dZ^d$ and their approximations by tori $\dT_n^d$. We will use the notations from \Cref{sec:finite}, with $G=\dT_n^d$. For notational simplicity we will replace $\dT_n^d$ by simply $n$. We will prove \Cref{thm:main} for weights with \textbf{good} distribution, and will remark later how this can be generalized. 
\begin{defn}\label{def:good}
Let $J$ be a random variable where $\{x:p(x)>0\}\supseteq (\beta,\infty)$ where $\beta\in \{-\infty\}\cup \dR$. Let $g_{z}$ denote the Radon-Nikodym derivative of the distribution $X+z$ with respect to that of $X$ for some fixed $z>0$, that is $g_{z}(x)=p(x-z)/p(x)$. We say that the distribution of $J$ is good if $p(\cdot)$ is continuous and there is some $\ga>1$ such that 
\begin{align}\label{eq:good}
C(z,\ga)=\int_{\gb}^{\infty}\bigl(g_{z}(x)\bigr)^{\ga}p(x)dx
\end{align}
is finite and continuous in $z$. 
 \end{defn}
To see that the goodness condition is rather broad, observe that it applies to Gaussian, Exponential and several families of power-law distributions, covering the spectrum in terms of tail. We do not know if \Cref{thm:main} holds for any nonatomic weight distribution, and we do not pursue such avenues in this article.

A collection of i.i.d.\ random variables coming from a good distribution is of course, almost surely generic. Given a matching $M$ of $\dZ^d$, its energy is now undefined. Nevertheless given two matchings $M$ and  $M'$ such that $M\Delta M'$ is finite, we can define the energy difference between $M$ and $M'$ as $H(M) - H(M'):= \sum_{x \in M \setminus M'} J_x - \sum_{y \in M'\setminus M}J_y$. We say $M$ is a \textbf{ground state of $\boldsymbol J$} if there is no other matching $M'$ such that $M\Delta M'$ is finite and $H(M) - H(M') >0$.

One way to construct ground states is to consider the a.s. unique ground state $M_n$ of  ${\boldsymbol J}_n := (J_x)_{x\in \Sigma_n}$. Let $\tilde M_n$ be the periodic extension of $M_n$ to all of $\dZ^d$. Then the collection $\{({\boldsymbol J}_n,\tilde{M}_n)\}_{n \ge 1}$ form a tight collection of random variables with range $\dR^{\gS} \times \Omega_{\dZ^d}$. For any subsequential weak limit $ (\boldsymbol J', M)$, the following properties hold:
\begin{itemize}
\item $\boldsymbol J'$ has the same law as $\boldsymbol J$, 
\item $M$ is a ground state matching with respect to $\boldsymbol J'$ almost surely, 
\item The law of $(\mvJ',M)$ is invariant with respect to translations of $\dZ^d$. 
\end{itemize}
The final item above holds simply because we took limits along tori. See \Cref{sec:generalization} about other possibilities.

Recall the definition of the excited states from \Cref{sec:excitedstates} and the notation $(M_{n,S,\xi})$ from \Cref{eq:metastate}. Let $\bf M_{n}$ denote the vector $$(M_{n,S,\xi})_{S \subset \Sigma, |S | <\infty, \xi \in \cV(S)}.$$ Of course, this includes $M_n$ as a marginal, corresponding to $S=\emptyset$. Let $\Delta H(n,S,\xi,\xi')$ be as in \eqref{eq:energy_diff_tight} (but with $G = \dT_n^d$). Let $\Delta \mvH$ be the vector $$ \Delta \mvH := (\Delta H(n,S,\xi,\xi'))_{S \subset \Sigma, |S| < \infty, \xi,\xi' \in \cV(S)}.$$ 
Here $\cV(S)$ is the set of valid configurations on $S$ in $\dZ^d$. We point out here that a configuration is valid in $\dT_n^d$ for a large enough $n$ if and only if it is valid in $\dZ^d$, see \Cref{lem:local_extension}. Since we are only interested in the limit, this slight abuse of notation is unambiguous.

Tightness of $\Delta \mvH$ is guaranteed by \eqref{eq:energy_diff_tight}, and hence we can take a further subsequential limit such that the vector $(\mvJ_{n},\mvM_{n}, \Delta \mvH_n)$ weakly converges to $(\mvJ,\mvM, \Delta \mvH)$. We emphasize that the limit at this point depends on the subsequence chosen, although it is not reflected in the notation. We will use different notations for different subsequential limit in later sections.

For later use, we list the following observations:  
\begin{itemize}
    \item The vector $(\bf J, \bf M)$ includes the limit $M$ as a marginal by taking $S = \emptyset$.
\item The flexibility vector $(F_n(x))_{x \in \sf \Sigma_n}$ is a continuous function of $(\mvJ_{n},\mvM_{n}, \Delta \mvH_n)$ and hence converges in law, call the limit $(F(x))_{x \in \sf \Sigma}$.
    
    \item For every $S, \xi$ as above, $M_{S, \xi}$ is independent of $J_S$ since this is true in finite volume. Using this fact and the continuity of the distribution of $J_x$,  we conclude that for every $x \in \gS$, 
    \begin{equation}
       \dP(F(x) >0 ) = 1  \label{eq:flex_0}.
    \end{equation}

    \item By an application of the Skorokhod representation theorem and Fatou's lemma, the infinite volume version of \Cref{eq:E_linear} is also true almost surely.
   Now observe that since the weights come from a continuous distribution, $$\{{\mvJ}: \Delta H(S, \xi, \xi') \neq  0 \text{ for all $S\subset \Sigma$}, |S | <\infty, \xi , \xi' \in \cV(S), \xi \neq \xi'\}$$ is a set of probability $1$. This allows us to define the following set:
    \begin{equation}
\cI(S,\xi) = \{J_S: \Delta H(S, \xi, \xi') < 0 \text{ for all }\xi' \in \cV(S) \} \label{eq:partition_infinite_volume}.
    \end{equation}
    which partitions a set of probability 1. Consequently, $(1_{J_S \in \cI(S, \xi)})_{S \subset \Sigma, |S |<\infty, \xi \in \cV(S)}$ is an almost surely 
   continuous function of $(\mv J, \mv M, \Delta \mvH)$.  Thus, \eqref{eq:metastate} also translates in the infinite volume to the following equality which holds almost surely:
\begin{equation}
    M = \sum_{\xi \in \cV(S) }M_{S, \xi}1_{J_S \in I(S, \xi)}\label{eq:metastate_infinite}
\end{equation}
\end{itemize}

It is a standard measure theoretic fact that the conditional law of ${\bf M}$ given an a.s. sample of ${\bf J}$ exists, call it $\mu_{{\bf J}}$ (the regular conditional probability exists for random variables on Borel spaces, see \cite{Durbook}). It is also standard that we can write ${\bf M} = \Phi({\bf J}, U)$ where $\Phi$ is a measurable function and $U$ is a Uniform random variable independent $\mvJ$ (since we are working on Borel spaces). Since $M_{S,\xi}$ is independent of $J_S$, we can actually write $M_S = \phi_{S,\xi}(J_{S^c}, U)$. Applying \eqref{eq:metastate_infinite}, we conclude that
\begin{align}
    M=\Phi({\bf J}, U) & = \sum_{\xi \in \cV(S) }\phi_{S,\xi}(J_{S^c}, U)1_{J_{S} \in \cI(S, \xi)}.  \label{eq:function_M}
\end{align}
For a given finite set $S$, by the joint convergence of $(\mvJ_{n},\mvM_{n})$, Corollary \ref{cor:modification} extends in a straightforward manner to the infinite lattice setting.

    

\begin{corollary}\label{cor:contmodification}
Suppose $S \subset \sf \Sigma_n$ is finite.
    Let $f:\dR^{S}\to \dR^{S}$ be a continuous function, and consider a new system of weights $\tilde \mvJ_n$ such that  $\tilde{\mvJ}_{n,S}=f(\mvJ_{n,S})$ and $\tilde{\mvJ}_{n,\Sigma_n \setminus S}=\mvJ_{n,\Sigma_n \setminus S}$. Similarly define $\tilde \mvJ$ such that  $\tilde{\mvJ}_{S}=f(\mvJ_{S})$ and $\tilde{\mvJ}_{\Sigma \setminus S}=\mvJ_{\Sigma \setminus S}$. Let $\tilde{M}_{n}$ denote the ground state with respect to $\tilde{\mvJ}_{n}$. Define $\tilde M =M_{S,\xi}$ where $\xi$ is such that $\tilde \mvJ_{S}\in \cI(S,\xi) $. In other words
    \[
    \tilde{M}=\sum_{\xi\in \cV(S)}M_{S, \xi}1_{\tilde{J}_{S}\in \cI(S,\xi)}.
    \]
    Then 
    $$
    (\mvJ_n,\tilde \mvJ_n, \tilde M_n, \mvM_n, \Delta \mvH_n) \xrightarrow[]{ }(\mvJ, \tilde \mvJ, \tilde M, \mvM, \Delta \mvH) 
    $$
   in law where the limit can be taken along any subsequence along which $( \mvJ_n, \mvM_n, \Delta \mvH_n) $ converges. Consequently, we can write
   \begin{equation}
       \tilde{M}=\Phi(\tilde {\mvJ},U)=\sum_{\xi\in \cV(S)}\phi_{S,\xi}(\tilde \mvJ_{S^{c}},U)1_{\tilde{J}_{S}\in \cI(S,\xi)}.\label{eq:function_u}
   \end{equation}
   for some measurable functions $\phi_{S,\xi}$, $\Phi$ and $U \sim $ Unif $[0,1]$.
\end{corollary}
The point of the above corollary is that we can write $\tilde M$ as the \emph{same} function of $(\tilde {\mv J}, U)$ as that in \eqref{eq:function_M} which will allow us to conclude absolute continuity of $(\tilde {\mvJ}, \tilde M)$ with respect to $(\mvJ, M)$ from absolute continuity of $\tilde \mvJ$ with respect to $ \mvJ$.
\begin{proof}
    Observe that $\tilde \mvJ_n$ is a continuous function of $\mvJ_n$. Furthermore, by \Cref{cor:modification}, $\tilde M_n$ satisfies
    \begin{equation*}
    \tilde{M_n}=\sum_{\xi\in \cV(S)}M_{n,S, \xi}1_{\tilde{J}_{n,S}\in \cI(n,S,\xi)}.
    \end{equation*}
    This ensures that $\tilde M_n$ is an a.s. continuous function of $(\mvJ, \mvM, \Delta \mvH)$.
 Combining these two facts, the corollary is immediate.   
\end{proof}

\begin{lem}\label{lem:perturb_as_event}
Let $S = (u_1,(u_1,u_2),\ldots, (u_{k-1},u_k), u_k)$, $E_S,u_j$  be as in \Cref{lem:pathmod}. Let $\cR(S, \eps)$ denote the following event
\begin{itemize}
    \item The edges of $S$ are alternatingly occupied by dimers corresponding to the ground state $M$,
    \item $(u_j, u_{j+1})$ is in $M$,
    \item $F((u_{k-1},u_k)) \ge \eps$
\end{itemize}
Let $z_{S}:=\max_{(u,v)\in E_{S}}|J_{(u,v)}-J_{u}-J_{v}|+V$ as in \eqref{eq:zs}, where $V$ is an independent uniform $[0,1]$ random variable. Define the weights $\mvJ^{S,\eps}$ as follows: \[J^{S,\eps}_{e}=J_{e}+z_{S}\cdot1_{\cR(S,\eps)}\text{ }\forall e\in E_{S},\text{ }J^{S,\eps}_{(u_{1},v)}=J_{(u_{1},v)}+\frac{\eps}{2}\cdot 1_{\cR(S,\eps)} \text{ }\forall v\sim u_{1},\]
all other weights remaining unchanged. Let $O^{S,\eps}$ denote relevant optimalities computed with respect to $\mvJ^{S,\eps}$. Let \[M^{S,\eps}:=\sum_{\xi \in \cV(A) }M_{A, \xi}1_{J^{S,\eps}_A \in I(A, \xi)},\] 
where $M_{A,\xi}$ is as in \eqref{eq:metastate_infinite} and $A=E_{S}\cup \{u_{1}\}$. Consider the events 
\begin{align*}
\cA_{1}(S,\eps)&:=\biggl\{O^{S,\eps}(u_1)=O(u_1)-\frac{\eps}{2}\biggr\},\\
\cA_{2}(S,\eps)&:=\biggl\{\text{All edges $e\in E_{S}$ are inaccessible} \biggr\}, \\\cA_{3}(S,\eps)&:=\bigl\{M\gD M^{S,\eps} =\emptyset \text{ or } M\gD M^{S,\eps} =Q\subset (u_1,(u_1,u_{2}, \ldots, (u_{p-1},u_p),u_p)) \text{ for some $p \le k-1$}\},
\end{align*}
where $Q$ is a simple finite path. The event 
\[
\cA(S,\eps):=\cA_{1}(S,\eps)\cap\cA_{2}(S,\eps)\cap\cA_{3}(S,\eps)
\]
holds with probability $1$ given $\cR(S,\eps)$, that is 
\[
\dP\bigl(\cA(S,\eps)|\cR(S,\eps)\bigr)=1. 
\]
 
\end{lem}

\begin{proof}
Pick $n$ large enough so that $S \subset \dT^d_{n}$. Let $\cR_n(S, \eps)$ be the same event as in $\cR(S, \eps)$ with $M$ replaced by $M_n$ and $F$ replaced by $F_n$.
Then define $\cA_{n,i}^{S, \eps}, \cA_n^{S, \eps}$ to be the same events as $\cA_i^{S, \eps}$, $i\in \{1,2,3\}, \cA^{S, \eps}$ but with $M$ replaced by $M_n$ and $M^{S, \eps}$ replaced by $M_n^{S, \eps}$ where $M_n^{S, \eps}$ is the ground state for $\mvJ^{S, \eps}$ in $\dT_n^d$. Merely as a consequence of the definitions, 
\[
\dP(\cA_{n,1}|\cR_{n}(S,\eps))=1.
\]
Again, merely as a consequence of definitions we know that if $\cR_{n}(S,\eps)$ holds, 
\[
J_{(u,v)}^{S,\eps}-J^{S,\eps}_{u}-J_{v}^{S,\eps}>0\text{ }\forall e=(u,v)\in E_{S}.
\]
Of course, this is exactly the condition for all $e\in E_{S}$ to be inaccesible and thus 
\[
\dP\bigl(\cA_{n,2}(S,\eps)|\cR_{n}(S,\eps)\bigr)=1.
\]
It is a corollary of Lemma \ref{lem:pathmod} that 
\[
\dP(\cA_{n,3}|\cR_{n}(S,\eps))=1.
\]
Thus,
\[
\dP\bigl(\cA_{n} (S,\eps)|\cR_{n}(S,\eps)\bigr)=1.
\]
It remains to be shown that these sequences of conditional probabilities do in fact converge to $\dP(\cA(S,\eps)| \cR(S, \eps))$. This convergence is guaranteed as we can verify that $\cR_{n}(S,\eps)$ is a continuity set for the probability measure induced by $(\mv J_n, \mv M_n, \Delta \mv H_n)$. The second coordinate is equipped with the discrete topology. Thus, it is not hard to see that the topological boundary $\partial \cR(S,\eps)$ is defined by the same conditions as $\cR(S,\eps)$, except with the terminal edge satisfying $F_{n}((u_{k-1},u_{k}))=\eps$, which we know holds with probability $0$ as the flexibility is a continuous random variable, and moreover remains a continuous random variable for any subsequential limit. Furthermore, by continuity of the transformation $\mvJ \mapsto \mvJ^{S, \eps}$ on $E_S$, \Cref{cor:contmodification} applies. By the Portmanteau theorem, for any weak limit, we now know that 
\[
\dP\bigl(\cA (S,\eps)|\cR(S,\eps)\bigr)=1. 
\]
This completes the proof.
\end{proof}
\section{Proof of Uniqueness}\label{sec:uniqueness}

Now, suppose we have two distinct subsequential limits $(\bf J,\bf M)$ and $(\bf J', \bf M')$. Since $\bf J$ and $\bf J'$ have the same distribution, let $\mu_{\bf J}$ denote the conditional law of $\bf M$ given an a.s. sample of $\bf J$ and let $\mu'_{\bf J}$ be the conditional law of $\bf M'$ given an a.s. sample of ${\bf J}$. Let $  {\bf M} \sim \mu_{\bf J}$ and $  {\bf M'} \sim \mu_{\bf J}'$ and assume they are independent (given $\mvJ$).  
The tuple $$\Xi:=(\bf J,  {\bf M},   {\bf M'})$$ defines a coupling between $\mvM$ and $\mvM'$.  By construction (since we took the weak limit of the ground states on tori), the law of $\Xi$ is invariant with respect to translations of $\dZ^d$.

 An \textbf{infinite path} is an infinite sequence  $(v_1,(v_1,v_2),(v_2,v_3),\ldots)$ such that $v_i$s are all distinct and $v_i \sim v_{i+1}$ for all $i \ge 1$. A \textbf{bi-infinite path} is a bi-infinite sequence of edges $(\dots, (v_{-2}, v_{-1}),(v_{-1}, v_0), (v_0, v_1), (v_2,v_3), \dots)$ such that  $v_i$s are all distinct. An analogue of \Cref{lem:pathsloops} in the infinite setup is that the symmetric difference of two matchings in $\dZ^d$ is a disjoint union of simple loops, finite simple paths, simple infinite and bi-infinite paths.

\begin{lem}\label{lem:ground_sym}
    Let $M$ and $M'$ be two ground states in $\dZ^d$ for a collection of almost surely generic weights ${\boldsymbol J}$. Then $M \Delta M'$ is a disjoint union of simple infinite and bi-infinite paths.
\end{lem}
\begin{proof}
    There cannot be a finite, simple path or a simple loop in $M \Delta M'$ as we can switch between the complimentary matchings of one such component and lose energy, which is impossible as both $M$ and $M'$ are ground states.
\end{proof}

 \begin{defn}
 A vertex $v$ is a \textbf{start point} (SP) of $M\Delta M'$ if $v \in M\Delta M'$.
 \end{defn}
 If $v$ is a start point then by \Cref{lem:ground_sym}, the connected component  containing $v$ is an infinite simple path with endpoint $v$. 
\begin{lem}\label{lem:BK}
$\dP( \text{There is a  starting point of $M\Delta M'$} )=0$. 
\end{lem}
\begin{proof}
The proof follows from a version of the Burton-Keane argument. We do not know that the law of $\Xi$ is ergodic with respect to translations of $\dZ^d$, nevertheless we can restrict to its ergodic components. Let $\tilde \Xi =  (\tilde {\bf J}, \tilde M, \tilde M')$ be such an ergodic component of $\Xi$. Let $$\alpha:=\dP ({\bf 0} \text{ is a starting point of $\tilde M \Delta \tilde M'$}).$$ Let $N_K$
denote the number of starting points of $M\Delta M'$ in $B_K$. 
Observe that by ergodic theorem, almost surely, $N_K \ge \frac\alpha2 (2K)^d$ for all large enough positive $K $. On the other hand, the number of vertices in $\partial B_K$ which the component containing a starting point intersects is at least 1, and no vertex in $\partial B_k$ belong to two different such components. Thus $|\partial B_K| \ge N_k$ for all $K \ge 0$. Combining the two observations, we see that $|\partial B_K| \ge \frac\alpha2 (2K)^d$ for all large enough $K$. Since $|\partial B_K |/K^d \to 0$
 as $K \to \infty$, $\alpha$ must be 0. Since the choice of the ergodic component was arbitrary, the lemma is proved.
\end{proof}

Because of \Cref{lem:BK,lem:ground_sym} if $M\Delta M'$ is non empty with positive probability then  it must be the case that it has contains one or more bi-infinite simple paths with positive probability. 
We say a bi-infinite simple path \emph{passes through a vertex }$x$ if one its edges has $x$ as one endpoint. In the next section, we show that the probability of such a  path passing through ${\bf 0}$ is 0, which will conclude the proof of \Cref{thm:main}.


\subsection{Nonexistence of bi-infinite paths}\label{sec:bi-infinite}
We know from \Cref{lem:BK} that $M \Delta M'$ consists of bi-infinite paths only.
The strategy now as described in \Cref{sec:outline}  is to do an absolutely continuous modification to the weights and arrive at a contradiction.

To achieve this, the strategy is to perform a modification of a sensitive vertex weight on the path that changes each ground state in a controlled manner. By prior edge weight modifications, we will force the changes to occur along the bi-infinite path. We will do this by first finding two edges of high flexibility along the path such that the high sensitivity vertex is sandwiched between them, and then by making any edge neighbouring the path which is not occupied inaccessible. Of course, such a region with a sensitive vertex sandwiched by high flexibility edges has to be shown to exist with positive probability.

Let $\cP$ be the set of matchings $m,m'$ of $\dZ^d$ such that $m \Delta m'$ contains an infinite path through the origin ${\bf 0}$. On the event $M \Delta M' \in \cP$, we refer to the path going through ${\bf 0}$ as $P$.
Recall the definition of optimality from \eqref{eq:opt}. We define the \emph{key quantity}
\begin{align}
{\sf m}:=\inf_{P \text{ passes through } v}O(v)1_{M \Delta M' \in \cP} + \infty1_{M \Delta M' \not \in \cP}.
\end{align}
Observe that on the event that $P$ passes through $v$, $O(v) >0$ almost surely. Indeed, otherwise, all the edges adjacent to $v$ are inaccessible for any ground state, which is a contradiction as $P$ passes through $v$. Thus ${\sf m} \ge 0$ almost surely.  
\begin{lem}\label{lem:MTP2}
 Assume $\dP(M \Delta M' \in \cP)>0$. For all  $\eps>0$ there exists a $c>0$, such that $\dP(M \Delta M' \in \cP,O({\bf 0})<{\sf m}+\eps)\ge c$
\end{lem}
\begin{proof}
Suppose there exists an $\eps$ such that $\dP(M \Delta M' \in \cP,O({\bf 0})<{\sf m}+\eps) =0$. But on the event $M \Delta M' \in \cP$, by definition of infimum, there exists a vertex $x$ such that $P$ passes through it and such that $O(x) <{\sf m}+\eps$. This contradicts translation invariance.
\end{proof}

\begin{lem}\label{lem:unique_min}
    $\dP(M \Delta M' \in \cP,\exists \text{ unique } v \in P, O(v) = {\sf m})=0$
\end{lem}
\begin{proof}
Send mass $1 $ from $x$ to $y$ if a bi-infinite path $P'$ in $M \Delta M'$ passes through both of them and $y$ is the unique vertex such that $O(y) = \inf_{ P' \text{ passes through $v$}} O(v)$. The mass out is at most 1, while the mass in is infinite on the event that ${\bf 0}$ is the unique vertex in $P$ with $O(\bf 0) = {\sf m}$. This is a contradiction if the latter event has positive probability, therefore the latter event must have 0 probability. The lemma follows by translation invariance.   
\end{proof}

Let us colour the edges belonging to $M$ as red
 and those belonging to $M'$ as blue. Recall that $P$ is composed of alternating red and blue edges. Furthermore, on the event $M \Delta M' \in \cP$, the path $P$ can be split at the origin into two single directional infinite paths one of them starting with a red edge, and another with a blue edge. Call the former the \textbf{red direction} and the latter the \textbf{blue direction} of $P$.  Let us denote  by $F_M(x)$ (resp. $F_{M'}(x)$) the flexibility of $x$ corresponding to $({\bf J}, \mvM)$ (resp. $(\mvJ,\mvM')$).

 Let $\cG =\cG(\eps)$ be the  following event:
\begin{itemize}
    \item $M \Delta M' \in \cP$ occurs.
    \item $O(\bf 0) <{\sf m}+\eps/4$.
    \item there exists infinitely many red edges $e$ in $P$ in the red direction with $F_M(e) \ge \eps$, and infinitely many blue edges $e'$ with $F_{M'}(e')\ge \eps$ in the blue direction. 
\end{itemize}
 \begin{lem}\label{lem:MTP1}
 Assume $\dP(M \Delta M' \in \cP)>0$. Then
There exists an $\eps_F>0$ such that $$\dP(\cG(\eps_F)) >0.$$
\end{lem}

In what follows, we do not need the full strength of the last item in the above event.
\begin{proof}

We already know from \Cref{lem:MTP2} that we can choose $\eps_1>0$ such that $$\dP(M \Delta M' \in \cP, O({\bf 0}) < {\sf m}+\eps_1) \ge c>0.$$ Let $\cE(\eps)$ be the event that there are infinitely many red edges $e$ in $P$ in the red direction with $F_M(e) \ge \eps$. We now claim that there exists $\eps_2>0$ such that 
\begin{equation}
    \dP(M \Delta M' \in \cP, O({\bf 0}) < {\sf m}+\eps_1, \cE^c(\eps_2)) < \frac{c}{10}.
\end{equation}
Let $\cF$ be the event that the red edge $e$ adjacent to the root has $F_M(e) \ge \eps$. Using \eqref{eq:flex_0}, we can choose $\eps$ small enough such that $\dP(\cF^c) < c/10$. Let $\eps_2$ be the choice of $\eps$. We now claim
\begin{equation}
    \dP(M \Delta M' \in \cP, O({\bf 0}) < {\sf m}+\eps_1, \cE^c(\eps_2), \cF) =0.\label{eq:finite}
\end{equation}
Let $\cB$ be the event inside \eqref{eq:finite} and we prove \eqref{eq:finite} by contradiction. To that end, assume $\dP(\cB) >0$. On $\cB $ there is a `last vertex' $v$ whose red direction does not contain any edge $e$ with $F_M(e) \ge \eps_2$  in the following sense. For $x \in V$, let $\cL(x)$ be the event that a bi-infinite path in $M \Delta M'$ passes through $x$, $O({\bf 0})<{\sf m}+\eps_1$ and the red edge $e$ adjacent to $x$ is the only edge in the red direction with $F_M(e) \ge \eps_2$. By translation invariance in law of $\Xi$, we conclude that $\dP(\cL({\bf 0})) >0$ as the probability $\dP(\exists x, \cL(x) \text{ occurs}) \ge \dP(\cB) >0$. 
 For any $u,v$ define a mass transport sending mass 1 from $u$ to $v$ if a bi-infinite path in $M \Delta M'$ passes through both $u$ and $v$, and $\cL(v)$ holds. Clearly, the mass out of the root is at most 1 by definition. On the other hand, on the event $\cL({\bf 0})$, ${\bf 0} $ gets mass $1$ from every vertex through which its blue direction passes, hence expected mass into $\bf 0$ is infinity. This contradicts mass transport principle, thereby proving \eqref{eq:finite}.
 Thus we get $$\dP(M \Delta M' \in \cP, O({\bf 0}) < {\sf m}+\eps_1, \cE^c(\eps_2)) = \dP(M \Delta M' \in \cP, O({\bf 0}) < {\sf m}+\eps_1, \cE^c(\eps_2), \cF^c)<\frac{c}{10}.$$ as desired.

 Exactly similar bounds for the 7 other combinations in the third item of $\cG$, and choosing $\eps_F$ to be the infimum of the $\eps$s chosen along with an union bound will  gives us
$$\dP(\cG(\eps_F)) \ge c- \frac{8c}{10} = \frac{c}{5}>0$$
as required.
\end{proof}

\begin{rem}
Both \Cref{lem:MTP1,lem:unique_min} can be proved by appealing to the amenability of $\dZ^d$, but we prefer to use mass transport principle as it avoids using amenability.
\end{rem}
\bigskip
\begin{figure}
    \centering
    \includegraphics[scale=0.6]{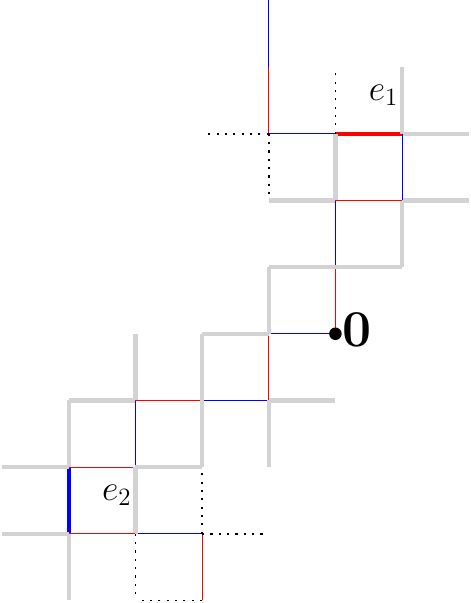}
    \caption{The set up for the event $\cG$. The edges in $E_Q$ are marked in gray. }
    \label{fig:path}
\end{figure}

On $\cG(\eps_F)$ denote by $e_1$ the closest red edge in the red direction with with $\min\{F_M(e_1)\} \ge \eps_F$ and similarly define $e_2$ to be the closest blue edge in the blue direction with $F_{M'}(e_2) \ge \eps_F$. Let $Q$ be the  set of vertices incident to $P$ between them. To be more precise, $Q$ consists of the endpoints of the component of $P \setminus \{e_1,e_2\}$ containing ${\bf 0} $. Note here that this component of $P$ containing ${\bf 0} $ might not have a single edge in which case $Q$ is the singleton ${\bf 0}$. Let $E_{Q}$ denote the set of edges adjacent to $Q$ except those adjacent to ${\bf 0}$.  

In what follows we will take a union bound over $Q =S$ for all possible potential finite paths  $S$ passing through the origin and consider perturbation over the set $E_{S}\cup \{\mv0\}$. This motivates the following definition. Now define $$
z_S = \max_{e = (u,v) \in E_{S}}\{|J_{(u,v)}-J_u - J_v|\}+V,$$
where $V$ is an independent uniform $[0,1]$ valued random variable just as in \eqref{eq:zs}. Consider a new collection of weights $ {\bf J}^{S,\eps}:=(J^{S, \eps}_x)_{x \in \Sigma}$ defined as follows:
\begin{equation}\label{eq:perturb_J}
   J^{S, \eps}_x = 
    \begin{cases}
       J_x +z_S \text{ if $x \in E_{S}$}\\
       J_{x} +\eps/2 \text{ if $x = {(\bf 0,v)},\text{ }v\sim \mv0$}  \\
     J_x \text{ otherwise.}
    \end{cases}
\end{equation}
Now define 
$$
\cG(\eps):= \cG \cap \{Q  =S\}
$$
and note that
$$
\cG(\eps) \subseteq \cup_{S} \cG(S, \eps)
$$
where the union runs over all finite subsets of $\sf E$. Note that by \Cref{cor:contmodification}, we can consider the perturbed ground states coupled with the original one:
\begin{align*}
    X&:= (\mvJ,\mvM,\Delta \mvH,{\bf J^{S,\eps}},M^{S, \eps})\\
    X'&:=(\mvJ,\mvM',\Delta \mvH',{\bf J^{S,\eps}},(M')^{S, \eps}).
\end{align*}
Finally, conditioned on ${\mvJ, \mvJ^{S, \eps}}$ we can independently sample the remaining coordinates of $X, X'$ to construct a coupling of $X,X'$:
\begin{equation} \label{Xi}
    \Xi:=(\mvJ, {\bf J^{S,\eps}}, \mvM,\Delta \mvH, \mvM',\Delta \mvH',M^{S, \eps},(M')^{S, \eps} )
\end{equation}

\begin{lem}\label{lem:abs_inf} Fix $S$, a finite path passing through $\mv0$ and $\eps>0$. Then   $({\bf J^{S,\eps}}, M^{S, \eps},(M')^{S, \eps})$ is absolutely continuous with respect to that of $({\bf J}, M,M')$.
\end{lem}
\begin{proof}
    The transformation that is applied to $\mvJ$ to obtain $\mvJ^{S, \eps}$ when restricted to $S$ is a continuous transformation, thereby \Cref{cor:contmodification} is applicable to each of $X$ and $X'$ defined above. Furthermore, since the ground states conditioned on ${\mvJ}, \mvJ^{S, \eps}$ are sampled independently, we can use two i.i.d.\ Uniform variables $U,U'$ in the formula \eqref{eq:function_u}. Applying this, we get
    \begin{align*}
    M^{S, \eps} = \Phi(\mvJ^{S, \eps}, U), \qquad & M = \Phi(\mvJ, U) \\
   (M')^{S, \eps} = \Phi'(\mvJ^{S, \eps}, U'), & \qquad M' = \Phi'(\mvJ, U') 
    \end{align*}
    for some measurable functions $\Phi, \Phi'$.
    It suffices then to verify that $({\mvJ}^{S, \eps}, U, U')$ is absolutely continuous with respect to $(\mvJ, U, U')$. This is standard, we postpone the proof of this to the appendix (see  \Cref{lem:abs}). 
\end{proof}
We now introduce a quick piece of notation: on $\cG(\eps_F)$, let $Q_{r} \subset Q$ be the simple, finite path starting at ${\bf 0}$ and ending at $e_1$ (excluding $e_1$) and define $Q_b \subset Q$ to be the simple, finite path starting at ${\bf 0} $ and ending at $e_2$ (excluding $e_2$).

\begin{lem}\label{lem:upper_bound_0}
Let $\cD_S$ be the event that 
\begin{itemize}
    \item  either there exists a  monomer in $M^{S, \eps_F}\Delta (M')^{S, \eps_F}$ or,
    \item $M^{S, \eps_F}\Delta (M')^{S, \eps_F} \in  \cP$ and $\inf_{P \text{ passes through $v$}} O^{S,\eps_{F}}(v)$ is achieved for some unique vertex through which $P$ passes where $O^{S,\eps_{}F}$ is the optimality is calculated with weights ${\mvJ^{S,\eps_{F}}}$.
\end{itemize}  
    $$
    \dP(\cG(S,\eps_F)) \le \dP(\cD_S  ) = 0.
    $$
\end{lem}
\begin{proof}

We work with the coupling $\Xi$ as in \eqref{Xi}.
   Let $\cR(Q_{r},\eps_{F})$ and $\cA(Q_{r}, \eps_F)$ be the events described in \Cref{lem:perturb_as_event} for $X$ and $\cR'(Q_{b},\eps_{F})$ and $\cA'(Q_{b}, \eps_F)$ be the same events for $X'$. Note that since $\cG(S,\eps_{F})\subset \cR(Q_{r},\eps_{F})$ and $\cG(S,\eps_{F})\subset \cR'(Q_{b},\eps_{F})$, we know that $\dP\bigl(\cA(Q_{r},\eps_{F})|\cG(S,\eps_{F})\bigr)=1$ and $\dP\bigl(\cA'(Q_{b},\eps_{F})|\cG(S,\eps_{F})\bigr)=1$. Thus, 
   \begin{equation}
       \dP(\cG(S, \eps_F)) = \dP(\cG(S,\eps_F), \cA(Q_{r},\eps_F), \cA'(Q_{b},\eps_F)).
   \end{equation}
   However, if all of $\cG(S,\eps_F),\cA(Q_r,\eps_F), \cA'(Q_b,\eps_F)$ occurs  and one of $M \Delta M^{S, \eps_F} =:R$ or $M' \Delta (M')^{S, \eps_F}=:R'$ is nonempty then a monomer must appear in $M^{S, \eps_F}\Delta (M')^{S, \eps_F}$. Indeed, $R \subset Q_r$ on $\cA(Q_{r},\eps)$ and $R'\subset Q_b$ on $\cA'(Q_{b},\eps)$. Thus $R$, if nonempty, will have a monomer at the the endpoint which is not ${\mv0}$. If both $R$ and $R'$ are empty, then the optimality of ${\bf 0}$ decreases to ${\sf m} - \eps/2$ for ${\bf J^S}$ whereas $O^{S,\eps_{F}}(w)$ for all $w$ outside $Q$ remains unchanged and hence are at least $\sf m$. Thus the infimum of the optimality of the vertices is achieved for some vertex in $Q$. Overall, $\cD_S$ has occurred. 

   The fact that $\dP(\cD_S)=0$ follows from the fact that the same event for $(\mv J, \mv M, \mv M')$ has probability 0 using \Cref{lem:BK,lem:unique_min}, and the absolute continuity result of \Cref{lem:abs_inf}.
\end{proof}

\begin{proof}[Proof of \Cref{thm:conditional_equality}]
    If $\dP(M \Delta M' \in \cP)>0$ then $\dP(\cG(\eps_F))>0$ using \Cref{lem:MTP1}. On the other hand, using \Cref{lem:upper_bound_0}, $$\dP(\cG(\eps_F) ) \le \sum_S \cG(S,\eps_F) =0$$
    where the sum is over all finite subsets $S \subset \sf E$. This is a contradiction, and hence $\cP(M \Delta M' \in \cP ) =0$. Combined with \Cref{lem:BK}, we see that $M\Delta M' = \emptyset$ almost surely, as desired. 
\end{proof}
 \section{Perturbation of ground states}\label{sec:perturbation}

In this section we prove \Cref{thm:perturbation}. The proof essentially follows the same ideas as in \Cref{thm:conditional_equality}, except in this case there can be finite components. 
Nevertheless an identical application of Burton and Keane (\Cref{lem:BK}) almost surely rules out the existence of starting points in $M \Delta M(p)$. Thus we are left to rule out existence of bi-infinite paths in $M \Delta M(p) $. Thus we let $\cP$ be the event as in \Cref{sec:bi-infinite} and assume  $\dP(M \Delta M(p) \in \cP)>0$. On $\cP$, let $P$ denote the path passing through ${\mv0}$ as in \Cref{sec:uniqueness}.

Let $O_{\mvJ}(v)$ denote the optimality of $v$ for the weights $\mvJ$. Define 
$${\sf m}_{\mvJ}:=\inf_{v \text{ passes through } P } O_{\bf J}(v).$$ The following lemma has identical proofs as \Cref{lem:MTP1,lem:unique_min}.
\begin{lem}\label{lem:perturb_properties}
    Assume $\dP(M \Delta M(p) \in \cP)>0$. For all $\eps>0$ there exists a $c>0$ such that 
    $$\dP(M \Delta M(p) \in \cP, O_{\mvJ}({\bf 0})<{\sf m}+\eps) \ge c.$$ Furthermore, $$\dP(M \Delta M(p) \in \cP,\exists \text{ unique } v \in P, O_{\mvJ}(v) = {\sf m}_{\mvJ})=0.$$
\end{lem}
As before, we consider the metastates by taking possibly subsequential limits:
\begin{equation*}
    (\mvJ, \mvM, \Delta \mvH), \qquad (\mvJ, \mvJ', \mvM(p), \Delta \mvH(p)).
\end{equation*}
This allows us to consider the flexibilities $F_{M}, F_{M(p)}$.

Let us color the edges of $M$ red and $M(p)$ blue and define the red direction and the blue direction of $P$ as before on the event $M \Delta M(p) \in \cP$. Define $\cG =\cG(\eps)$ to be the  following event:
\begin{itemize}
    \item $M \Delta M(p) \in \cP$,
    \item $O_{\bf J}(\bf 0) <{\sf m}_{\mvJ}+\eps/4$.
    \item there exists infinitely many red edges $e$ in $P$ in the red direction with $F_M(e) \ge \eps$, and infinitely many blue edges $e'$ with $F_{M(p)}(e')\ge \eps$ in the blue direction. 
\end{itemize}
Under the assumption $\dP(M \Delta M(p) \in \cP)>0$, there exists an $\eps_F>0$ such that 
\begin{equation}
    \dP(\cG(\eps_F))\ge c. \label{eq:pos_prob_event}
\end{equation}
On the event $\cG(\eps_F)$, let $Q$, $E_Q$, $\cG(\eps_F, S)$ be as in \Cref{sec:uniqueness}. Now define $$
z_S = \max_{e = (u,v) \in S}\{|J_{(u,v)}-J_u - J_v|, |J'_{(u,v)}-J'_u - J'_v|\} + V.$$
where $V \sim $Uniform$[0,1]$ independent of the rest.
Consider a new collection of weights $ {\bf J}^{S,\eps}:=(J^{S, \eps}_x)_{x \in \Sigma}$ defined as follows:
\begin{equation}
   J^{S, \eps}_x = 
    \begin{cases}
       J_x +z_S \text{ if $x \in S$}\\
       J_{x} +\eps/2 \text{ if $x = ({\bf 0},v),\text{ }v\sim \mv0$}  \\
     J_x \text{ otherwise.}
    \end{cases}
\end{equation}
and ${\bf J'}^{S, \eps}$ as 
\begin{equation}
   (J')^{S, \eps}_x = 
    \begin{cases}
       J'_x +z_S \text{ if $x \in S$}\\
     J'_x \text{ otherwise.}
    \end{cases}
\end{equation}
which allows us to define $(J(p))^{S, \eps}$ in the obvious manner. We can now define
\begin{align*}
    M^{S, \eps}&=\sum_{\xi\in \cV(S)}M_{S, \xi}1_{J^{S, \eps}\in \cI(S,\xi)}\\
    (M(p))^{S, \eps}&=\sum_{\xi\in \cV(S)}M_{S, \xi}1_{(J(p))^{S, \eps}\in \cI(S,\xi)}
\end{align*}
and exactly as in \Cref{lem:abs_inf}, we get that $(\mvJ, \mvJ(p), M, M(p))$ is absolutely continuous with respect to that of $((\mvJ)^{S, \eps}, (\mvJ(p))^{S, \eps}, M^{S, \eps}, (M(p))^{S, \eps})$. Since the perturbation makes sure that the edges in $S$ are made inaccessible for both $M^{S, \eps_F}$, $(M(p))^{S, \eps_F}$, we conclude in the same way  that on the event $\cG(S)$,
\begin{itemize}
    \item either there is a monomer in $M^{S, \eps} \Delta (M(p))^{S, \eps})$, or,
    \item $M^{S, \eps} =M$ and $(M(p))^{S, \eps} = M(p)$ and infimum of $O_{(\mvJ)^{S, \eps}}(v)$ over all vertices $v$ through which $P$ passes is achieved for some vertex.
\end{itemize} 
By absolute continuity and \Cref{lem:perturb_properties}, the latter event has probability 0. This renders $\dP(\cG(\eps))=0$ which is a contradiction to \eqref{eq:pos_prob_event}. Thus $M \Delta M(p) $ has finite components almost surely, as desired.

Now let us show that $M\Delta M(p) \to \emptyset$ in law as $p \to 0$. To that end, it is enough to show that for any sequence $p_k \to 0$, there exists a subsequence $p_{k_l}$ along which $\dP(\mv0 \in M\Delta M(p_{k_l})) \to 0$ as $l \to \infty$.

To that end, using compactness choose a sequence $p_{k_l}$ such that $(\mvJ(p_{k_l}), M(p_{k_l}))$ converges in law as $l \to \infty$. Call the limit $(\mvJ, M(0+))$ where we use the obvious fact that the marginal of $\mvJ(p_{k_l})$ converges to $\mvJ$ in law.
Observe that there exists a slowly growing function $\eta: \dN \to \{k_l\}_{l \ge 1}$  such that $(\mvJ(p_{\eta(n)}), M_n(p_{\eta(n)})) \to (\mvJ, M(0+))$ in law. Furthermore, $(\mvJ, M(0+))$ is translation invariant in law and $M(0+)$ is a ground state of ${\mvJ}$.  Thus we are exactly in the setup of the proof of \Cref{thm:conditional_equality} except we sample $M, M(0+)$ conditionally independently of ${\mvJ}$. Since $M(0+)$ is obtained as a weak limit, the results of \Cref{sec:finite} applies, particularly \Cref{lem:pathmod,lem:perturb_as_event}. Running the same argument as in \Cref{sec:uniqueness}, we obtain $M\Delta M(0+) = \emptyset$ almost surely, as desired.

\section{Central Limit Theorem}\label{sec:CLT}
In this section, we will be working under the  additional assumption that $\E|J_{x}|^{4+\gd}<\infty$ for some $\gd>0$. The Central Limit Theorem for the ground state energy is obtained via an application of Chatterjee's powerful method for normal approximation \cite{C_Normal_08}. {While the ground state $M_{n}$ can be sensitive to perturbation of the weights, the ground state energy $H(M_{n})$ is quite robust in the following sense.} Let $x\in \gS$. If we modify the weight of $x$, that is we consider the ground state with the new weight $J_{x}'=J_{x}+\eps$, then we have that $|H(M'_{n})-H(M_{n})|<\eps$. In particular, the ground state energy is a $1-$Lipschitz function of the weights $\mvJ$. The uniqueness of the ground state enables an adaptation of Lam and Sen's application of Chatterjee's method for the central limit theorem of the free energy in the zero temperature setting considered here. The following definitions for independent perturbations of our weights $\mvJ$ will be useful.  
\begin{enumerate}[(a)]
    \item $\mvJ^{S}$ is the random vector obtained when all the weights corresponding to $S\subset \gS$, that is $\{J_{x}\}_{x\in S}$ are replaced by independent copies $\{J_{x}'\}_{x\in S}$. 
    $\mvJ^{x}$ is the same random vector defined above for singletons $S=\{x\}$. 
    \item For a measurable function $g:\dR^{\gS_{n}}\to \dR$, $\partial_{x}g:=g(\mvJ)-g(\mvJ^{x})$. 
    \item For $S\subset V_{n}\sqcup E_{n}$ not containing $x$, $\partial_{x}g^{S}(\mvJ):=g(\mvJ^{S})-g(\mvJ^{S\cup\{x\}})$.
\end{enumerate}
Chatterjee's result, as stated and used in \cite{SL_23}, bounds the Kolmogorov-Smirnov distance $d_{KS}(\cdot, \cdot)$ between a scaled and centered version of $g$ and a standard Gaussian random variable $Z$ as follows: 
\begin{align}\label{eq:chatterjee}
d_{KS}\left(\frac{g(\mvJ)-\E g(\mvJ)}{\gs_{g}},Z\right)\leq \frac{\sqrt{2}}{\gs_{g}}\left(\sum_{x,y}c(x,y)\right)^{1/4} +\frac{1}{\gs_{g}^{3/2}} \left(\sum_{x}\E |\partial_{x}g|^{3}\right)^{1/2}, 
\end{align}
where 
\begin{align}
\gs_{g}^{2}:=\var(g) \text{ and } c(x,y):=\max_{S,T}\cov(\partial_{x}g\cdot \partial_{x}g^{S},\partial_{y}g \cdot \partial_{y}g^{T}),
\end{align}
for all $S$ not containing $x$ and all $T$ not containing $y$. The function $g$ for our purposes is $H(M_n)$ where recall that $M_n$ is the a.s. unique ground state. 

\subsection{Variance Bounds }A variance lower bound is key to the application of \eqref{eq:chatterjee}, it is a necessary ingredient in showing that the error terms vanish as $n\to \infty$. 

\begin{lem}\label{lem:varb}
There exist constants $C_{1}$ and $C_{2}$ such that 
\[
C_{1}n^{d}\leq \var(H(M_{n}))\leq C_{2}n^{d}. 
\]
\end{lem}
As is usually expected, the lower bound is less straightforward to obtain as compared to the upper bound. Recall from \eqref{eq:flexibility} the definition of the transition point for $x\in \mathsf{V}\sqcup \mathsf{E}$, $K_{n,x}$. By definition, 
\[
\{K_{n,x}>J_{x}\}=\{x\in M_{n}\} \text{ and } \{K_{n,x}<J_{x}\}=\{x\notin M_{n}\}. 
\]
The following lemma is obvious from the convergence of $K_{n,x}$ and its independence with respect to $J_{x}$. However, the following quantitative version of the proof is instructive. 
\begin{lem}\label{lem:varlbaux}
Let $x\in \gS $. Let $J_{x}$ be the weight associated to $x$ and $J_{x}'$ be an independent copy. Then $\dP(\max\{J_{x},J_{x}'\}<K_{n,x})>c_{1}>0$ uniformly in $n$. 
\end{lem}

\begin{proof}
The probability of the event $\{x\in M_{n}\}$ admits entirely local bounds which are independent of $n$. When $x$ is a vertex, observe: 
\begin{align}
\cA_{x}:=\{O(x)<0\} =\{\max_{u:u\sim x} \left(J_{u}+J_{x}-J_{u,x}\right)<0\}\subseteq\{x\in M_{n}\}.
\end{align}
Likewise, if $x = (u,v)$ is an edge, 
\begin{align*}
\cB_{x}:=\{J_{u,v}<\min_{t\sim u, w\sim v}(J_{t,u}-J_{t}+J_{v,w}-J_{w},J_{v,w}-J_{w}+J_{u},J_{t,u}-J_{t}+J_{v},J_{u}+J_{v})\}\subseteq \{x\in M_{n}\}. 
\end{align*}

We provide the proof for the vertex case as it is more convenient to write, the same principle applies to the edge case. If $\cA_{x}'$ is the analogous event defined with only $J_{x}$ replaced with $J'_{x}$, then \[
\dP(\cA_{x}\cap \cA_{x}')\leq \dP(J'_{x},J_{x}<K_{n,x}).
\]
Since all our weights have continuous distribution with full support, it follows that $\exists c_{1}>0$ such that $\dP(\cA_{x}\cap \cA_{x}')>c_{1}$. Clearly, since the event $\cA_{x}$ depends only on the immediate neighbourhood of $x$, $c_{1}$ is independent of $n$. 
\end{proof}
Note that the same method used to prove Lemma \ref{lem:varlbaux} can be used to show that the event $\{J_{x}<K_{n,x}<J_{x}'\}$ has strictly positive probability, uniformly bounded away from zero in $n$. Restriction to the event described in Lemma \ref{lem:varlbaux} enables an easier computation of the variance lower bound.
\begin{proof}[Proof of Lemma \ref{lem:varb}]
The upper bound follows via a straightforward application of the Efron-Stein inequality. Let $x\in {\mathsf \gS_{n}}$, and consider the ground state that rises from independent replacement of the weight $J_{x}$ with an independent copy $J_{x}'$. Since the energy of the ground state is a $1-$Lipschitz function of the weights, we know that 
\[
|H_{n}(M_{n})-H_{n}(M_{n}^{x})|\leq |J_{x}-J'_{x}|.
\]
Thus, by the Efron-Setin inequality,
\[
\var(H(M_{n}))\leq \frac{1}{2}\sum_{x\in \gS_{N}}\E (H_{n}(M_{n})-H_{n}(M_{n}^{x}))^{2}\leq \frac{1}{2}\sum_{x\in \gS_{n}}\E (J_{x}-J'_{x})^{2}, 
\]
which immediately implies that there is a constant $C_{1}$ such that
\[
\var{H(M_{n})}\leq C_{1}n^{d}.
\]
For the lower bound, we will employ a martingale approach. Let us enumerate the vertices in $\mathsf V_{n}$. Consider the sequence of filtrations $\cF_{j}:=\gs(\{J_{1},J_{2},J_{3}\ldots J_{j}\})$. We use this filtration to define a martingale
\[
H_{n,j}:=\E(H(M_{n})|\cF_{j}),
\]
and corresponding martingale differences 
\[
\gD H_{n,j}:=H_{n,j+1}-H_{n,j}. 
\]
Adding up the martingale differences and using the fact that the variance is always bounded below by the variance of the conditional expectation, we get
\[
\var{H(M_{n})}\geq \var{H_{n,n^{d}}}=\sum_{j=0}^{n^{d}-1}\var{\gD H_{n,j}}.
\]
Let $J_{j+1}'$ denote an independent copy of $J_{j+1}$, let $M_{n}^{(j+1)}$ be the ground state obtained after the replacement of $J_{j=1}$. We may express the martingale difference in terms of independent replacement as
\[
{\gD H_{n,j}}= \E(H(M_{n})|\cF_{j+1})-\E(H(M_{n})|\cF_{j})=\E(H(M_{n}^{})-H(M_{n}^{(j+1)})|\cF_{j+1}). 
\]
Thus, 
\[
\var(\gD H_{n,j})=\E (\E(H(M_{n}^{})-H(M_{n}^{(j+1)})|\cF_{j+1})^{2})
\]
By a combination of Jensen's inequality and the tower property, 
\begin{align*}
\E (\E(H(M_{n})-H(M_{n}^{(j+1)})|\cF_{j})^{2})\geq& \E (\E(H(M_{n})-H(M_{n}^{(j+1))})|J_{j})^{2}) \\
&=\frac{1}{2}\E (\E(H(M_{n}^{})-H(M_{n}^{(j+1)})|J_{j+1},J_{j+1}')^{2})
\end{align*}
We may further bound below by restricting to the event $\{J_{j+1}'<K_{n,j+1}\} \cap\{J_{j+1}<K_{n,j+1}\}$.
\begin{align}
\E (H(M^{(j+1)}_{n})-H(M_{n}))^{2}\geq \E (H(M^{(j+1)}_{n})-H(M_{n}))^{2}\mv1_{J_{j+1},J_{j+1}'<K_{n,j+1}}
\end{align}
On this event,
\begin{align}
(H(M_{n}^{(j)})-H(M_{n}))^{2}\mv1_{J_{j}',J_{j}<K}=(J_{j}'-J_{j})^{2}\mv1_{J_{j}',J_{j}<K_{n,j}}
\end{align}
Note that $J'=J$ with probability $0$. If $\E (J'-J)^{2}\mv1_{\cA}=0$, then it must be that $\dP(\cA)=0$. Thus, $\E(J_{j}'-J_{j})^{2}\mv1_{J_{j}',J_{j}<K(j)}>C_{2}$ for some $C_{2}>0$ by Lemma \ref{lem:varlbaux}. In turn, this tells us that $\var(\gD M_{n,j})>\eps$, and finally on adding up the martingale differences, 
\[
\var{M_{n}}>C_{2}n^{d}. 
\]
\end{proof}
There are two error terms in \eqref{eq:chatterjee} which we need to show vanish in the limit. The variance lower bound is adequate to show that the latter of the two vanishes when $g(\mvJ)=H(M_{n}(\mvJ))$. Indeed, thanks to the finite fourth moment assumption,  
\begin{align}
\frac{1}{(\var H(M_{n}))^{3/4}}\left(\sum_{x\in V_{n}\sqcup E_{n}}\E |\partial_{x}H(M_{n})|^{3}\right)^{1/2}\leq \frac{1}{C_{1}n^{3d/4}}\left(\sum_{x} |J_{x}'-J_{x}|^{3}\right)^{1/2}\leq \frac{Cd^{1/2}n^{d/2}}{C_{2}n^{3d/4}}.
\end{align}

\subsection{Transition Point Convergence} We now upper bound the first term on the right hand side of \eqref{eq:chatterjee}.  
\begin{lem}\label{lem:CLTerrbd}
There exists a constant $C$ and $\eps(n,R)$ which decays to $0$ as $R\to \infty$ and $n\to \infty$ where $R<<n$, such that for every $R>1$ and $c(x,y)$ as in \eqref{eq:chatterjee}, 
\[
\sum_{x,y} c(x,y)\leq C(\norm{J}_{4+\gd}^{4})\cdot \left(R^{d}n^{d}+n^{2d}\eps(R,n)\right).
\]
\end{lem}
We will be defining the sequence $\eps(R,n)$ to be the deviation in a certain sense, of a ground state configuration that arises from the restriction to a box $B_{R}$. We define $M_{n,R,x}$ to be the ground state of the monomer dimer model defined on $B_{R}(x)\cap {\mathsf \gS_{N}}$ with periodic boundary conditions for $R<n$. Observe that 
\begin{align}\label{eq:replacement1}
\partial_{x}H(M_{n})= (J_{x}'-J_{x})\mv1_{J_{x},J_{x}'<K_{n,x}}&+(K_{n,x}-J_{x})\mv1_{J_{x}<K_{n,x}<J_{x}'} \\ &+(J_{x}'-K_{n,x})\mv1_{J_{x}'<K_{n,x}<J_{x}}.
\end{align}
Define 
\begin{align} \label{eq:reperrorbd}
E_{x,n,1}:=\frac{K_{n,x}-J_{x}}{J_{x}'-J_{x}}\mv1_{J_{x}<K_{n,x}<J'_{x}} \text{ and } E_{x,n,2}:=\frac{J_{x}'-K_{n,x}}{J_{x}'-J_{x}}\mv1_{J_{x}'<K_{n,x}<J_{x}}.  
\end{align}
We define $E_{x,R,1}$ and $E_{x,R,2}$ analogously using $K_{R,x}$ instead of $K_{n,x}$. Observe that $|E_{x,\cdot,\cdot}|\leq 1$ almost surely, and further 
\[
|E_{x,n,\cdot}-E_{x,R,\cdot}|\leq 2.
\]
In extracting a bound for the $c(x,y)$, we will be replacing $\partial_{x} H(M_{n})$ with its corresponding local version, and then show that the error of replacement is small. In particular, we seek a bound on $\partial_{x}H(M_{n})-\partial_{x}H(M_{n,R})$. We may express \eqref{eq:replacement1} in terms of the random variables defined in \eqref{eq:reperrorbd}, as
\begin{align*}
\partial_{x}H(M_{n})=(J'_{x}-J_{x})\left(\mv1_{J_x,J_x'<K_{n,x}}+E_{x,n,1}+E_{x,n,2}\right).
\end{align*}
Thus,
\begin{align*}
\partial_{x}H(M_{n})-\partial_{x}H(M_{R,x})=(J'_{x}-J_{x})\biggl((\mv1_{J_x,J_x'<K_{n,x}}-\mv1_{J_x,J_x'<K_{R,x}})+(E_{x,n,1}-E_{x,R,1})+(E_{x,n,2}-E_{x,R,2})\biggr)
\end{align*}
For convenience, we define 
\begin{align}
A_{1}(n,R,x) &:= \mv1_{J_{x},J_{x}'<K_{n,x}}-\mv1_{J_{x},J_{x}'<K_{R,x}}\label{eq:a1nr} \\
A_{2}(n,R,x) &:= E_{x,n,1}-E_{x,R,1}\nonumber \\ 
A_{3}(n,R,x) &:= E_{x,n,2}-E_{x,R,2}\nonumber. 
\end{align}
To reiterate, $|A_{1}|,|A_{2}|,|A_{3}|\leq 2$ almost surely. By H\"older's inequality, with $p=1+\frac{\gd}{4}$ and $q=1+\frac{4}{\gd}$,
\begin{align}\label{eq:4+gd}
\E|\partial_{x}H(M_{n})-\partial_{x}H(M_{R})|^{4}\leq \left(\E(J_{x}'-J_{x})^{4+\gd}\right)^{4/(4+\gd)}\cdot \left(\E|A_{1}+A_{2}+A_{3}|^{(16+4\gd)/\gd}\right)^{\gd/(4+\gd)}. 
\end{align}
The reason for the marginally higher moment requirement as compared to \cite{SL_23} is now apparent, the finite temperature version is appropriately ``smoothed out'', enabling the use of $\gd=0$, and consequently $q=\infty$. We  define 
\begin{align}
\eps(n,R):=\left(\E|A_{1}+A_{2}+A_{3}|^{(16+4\gd)/\gd}\right)^{\gd/(4+\gd)}. 
\end{align}
Proving the CLT comes down to showing that $\eps(n,R)$ vanishes in an appropriate sense. 
\begin{lem}\label{lem:Kcon}
As $n,R\to \infty$, for every $\gd>0$,
\[
\dP(|K_{n,x}-K_{R,x}|>\gd,\max\{K_{n,x},K_{R,x}\}>\gb)\to 0.
\]
where $\gb$ is the lower bound of the support of $J$.  
\end{lem}
\begin{proof}

By \Cref{thm:main}, we know that the triple $(\mvJ, \mvM_{n},\mvM_{R})$ converges in law to $(\mvJ, \mvM,\mvM)$ where $\mvM $ is a measurable function of $\mvJ$.  Thus for any $x\in {\mathsf\gS}$,
\begin{align}\label{eq:Kcon1}
\dP\biggl(J_{x}\in \bigl(\min\{K_{n,x},K_{R,x}\},\max\{K_{n,x},K_{R,x}\}\bigr)\biggr)=\dP(x\in M_n\gD M_{R})\to 0
\end{align}
as $n,R \to \infty$, where the symmetric difference is appropriately defined via periodic extension. Define
\begin{align}
\cA_{n,R,\gd}:=\{J_{x}\in (K_{n,x},K_{R,x}),\text{ }|K_{n,x}-K_{R,x}|>\gd,\text{ }\max\{K_{n,x},K_{R,x}\}>\gb\},
\end{align}
and for $M\in \dR$ and $y>0$, define 
\[
\cA_{n,R,\gd,M,y}:=\cA_{n,R,\gd}\cap \{-M\leq K_{n,x},K_{R,x}\leq M\}\cap\bigl\{\max\{K_{n,x},K_{R,x}\}>y+\gb\bigr\}.
\]
Now, using the independence of $J_{x}$ and $(K_{n,x},K_{R,x})$, we can estimate 
\begin{align*}
\dP(\cA_{n,R,\gd,M,y})\geq \dP\biggl(|K_{n,x}-K_{R,x}|>\gd,\text{ }\max\{K_{n,x},K_{R,x}\}\in(\gb+y,M),\text{}\min\{K_{n,x},K_{R,x}\}>-M\biggr)\cdot P_{M,y} 
\end{align*}

where 
\[
P_{M,y}=\min\bigl\{p(z):z\in [\max\{\gb+y,-M\},M]\bigr\}
\]
and $p(z)$ is the density of $J$. Observe that $p(z)>0$ on $(\gb,\infty)$ by hypothesis, and thus $P_{M,y}>0$ by continuity. Thus, 
\[
\dP\biggl(|K_{n,x}-K_{R,x}|>\gd,\text{ } \gb+y<\max\{K_{n,x},K_{R,x}\}\leq M,\text{ }\min\{K_{n,x},K_{R,x}\}>-M\biggr)\leq \frac{\dP(\cA_{n,R,\gd, M,y})}{P_{M,y}}.
\]
Now, let $\eps>0$ be fixed. Since $K_{n,x}$ and $K_{R,x}$ are both tight, we may find an $M_{\eps}$ such that $\dP(\max\{K_{n,x},K_{R,x}\}>M_{\eps})<\eps$. Further, since the sequence of sets $\cA_{n,R,\gd,M,y}$ is nested and increasing as $y\downarrow 0$, we may also find a $y_{\eps}$ sufficiently small such that $\dP(\cA_{n,R,\gd}\setminus \cA_{n,R,\gd,y_{\eps}})<\eps$ Thus,
\[
\dP(|K_{n,x}-K_{R,x}|>\gd, \max\{K_{n,x},K_{R,x}\}>\gb)\leq \frac{\dP(\cA_{n,R,\gd,M_{\eps},y_{\eps}})}{P_{M_{\eps},y_{\eps}}}+2\eps.
\]
Next, by \eqref{eq:Kcon1}, we know that for $n,R$ sufficiently large, $\dP(\cA_{n,R,\gd,M_{\eps},y_{\eps}})\leq P_{M_{\eps},y_{\eps}}\cdot\eps$. Thus, for every $\gd>0$ and $\eps>0$, for $n,R$ sufficiently large, 
\[
\dP(|K_{n,x}-K_{R,x}|>\gd,\max\{K_{n,x},K_{R,x}\}>\gb)<3\eps. 
\]
If $-\infty<\gb$ and $p(\gb)>0$, the step involving $y_{\eps}$ can be avoided. 
\end{proof}
 
\begin{corollary}\label{cor:epszero}

Let $R(n)$ be a sequence increasing to infinity such that $R/n$ is $o(1)$. As $n\to \infty$, 
\[
\E(|A_{1}|^{p})+\E(|A_{2}|^{p})+\E(|A_{3}|^{p})\to 0
\]
for all $1\leq p<\infty$. Thus, 
\[
\eps(n,R(n))\to 0. 
\]
\end{corollary}
\begin{proof}
For $i\in \{1,2,3\}$, 
\[
A_{i}(n,R)\cdot \mv1_{\max\{K_{n,x},K_{R,x}\}\leq\gb}=0.
\]
It follows from definition (see \eqref{eq:a1nr}) that 
\[
|A_{i}(n,R)|\leq 2\cdot \mv1_{\max\{K_{n,x},K_{R,x}\}>\gb}. 
\]
We may now apply the dominated convergence theorem and Lemma \ref{lem:Kcon} to conclude that  
\[
\E|A_{i}(n,R)|^{p}\to 0
\]
as $n,R(n)\to \infty$. The decay of $\eps(n,R(n))$ follows from Jensen's inequality. 
\end{proof}
\begin{corollary}\label{cor:decay_corr}
Let $x_{1},x_{2}\in {\mathsf \gS_{n}}$, such that $d(x_{1},x_{2})\geq 2R$. Then the following bound on the covariance holds:
\begin{align*}
\cov\left(\mv1_{x_{1}\in M_{n}},\mv1_{x_{2} \in M_{n}}\right)\leq 8\E|A_{1}(n,R)|+4\E(A_{1}(n,R))^{2}\leq 12\eps(n,R)^{(4+\gd)/\gd},
\end{align*}
where $A_{1}(n,R)$ is defined in \eqref{eq:a1nr}. Consequently, 
$$
\cov\left(\mv1_{x_{1}\in M_{n}},\mv1_{x_{2} \in M_{n}}\right) \to 0
$$
as $n,R \to \infty$ with $R = o(n)$.
\end{corollary}

\begin{proof}
Using the bilinearity of the covariance, we can expand the covariance as follows: 
\begin{align*}
\cov\left(\mv1_{x_{1}\in M_{n}},\mv1_{x_{2} \in M_{n}}\right) &=\cov\left(\mv1_{x_{1}\in M_{R}},\mv1_{x_{2} \in M_{R}}\right)+\cov(A_{1}(n,R,x_{1}),\mv1_{x_{2}\in M_{R}})\\ &+\cov(\mv1_{x_{1}\in M_{R}},A_{1}(n,R,x_{2}))+\cov(A_{1}(n,R,x_{1}),A_{2}(n,R,x_{2})).
\end{align*}
Note, since $d(x_{1},x_{2})>2R$, the tori of side length $R$ centered at $x_{1}$ and $x_{2}$ are disjoint. Therefore, the ground states defined on the respective boxes are independent, since the corresponding weights are independent. This tells us
\[
\cov(\mv1_{x_{1}\in M_{R}},\mv1_{x_{2}\in M_{R}})=0. 
\]
The following bound is straightforward since $\mv1_{x\in M_{R}}$ is bounded above by $1$;
\begin{align*}
\cov(A_{1}(n,R,x_{1}),\mv1_{x_{2}\in M_{R}})\leq 4\E|A_{1}(n,R)|. 
\end{align*}
Of course an analogous bound follows when the roles of $x_{1}$ and $x_{2}$ are swapped. Next, as a consequence of the Cauchy-Schwarz inequality, we find that 
\[
\cov(A_{1}(n,R,x_{1}),A_{2}(n,R,x_{2}))\leq 4\E (A_{1}(n,R))^{2}. 
\]
The final bound follows from Jensen's inequality.  
\end{proof}
\subsection{Proof of CLT} 

We emphasize now that $\partial_{x}H(M_{n})$ and $\partial_{x}H(M_{n})^{S}$ are identically distributed, so any moment calculations directly carry over. We will write, for convenience
\begin{align}
G_{R,S,x}:=\partial_{x}H(M_{n})^{S}-\partial_{x}H(M_{R})^{S}. 
\end{align}
{and \[G_{R,x} = G_{R, \emptyset, x}.\]}
\begin{proof}[Proof of Lemma \ref{lem:CLTerrbd}]
The primary ingredient in this proof is a double application of the Cauchy-Schwarz inequality. Note, for $S$ not containing $x$ and $T$ not containing $y$, we need to evaluate the following covariance:
\[
\cov\biggl (\partial_{x}H(M_{R,x})+G_{R,x})\cdot (\partial_{x}H(M_{R,x})^{S}+G_{R,S,x}), (\partial_{y}H(M_{R,y})+G_{R,y})\cdot (\partial_{y}H(M_{R,y})^{T}+G_{R,T,y})\biggr). 
\]
There are $16$ terms that arise on expansion of the product. First, observe that for any monomial that contains a $G$ term, that is $\cov(XY,ZW)$ where atleast one of $X,Y,Z,\text{ or }W$ is one of the error terms $G$, we have that 
\[
\cov(XY,ZW)\leq C\eps(n,R).
\]
We obtain this bound by applying the Cauchy-Schwarz inequality twice and using \eqref{eq:4+gd}. We still need a bound on the expectation of the unique monomial containing all non $G$ terms;
\begin{align}\label{lem:bigexpectation}
\cov(\partial_{x}H(M_{R,x})\partial_{x}H(M_{R,x})^{S},\partial_{y}H(M_{R,y})\partial_{y}H(M_{R,y})^{T}). 
\end{align}
In this case, note that the expectation vanishes unless $y\in B_{R}(x)$, or equivalently $x\in B_{R}(y)$. This is because if $x$ and $y$ are separated by a distance greater than $R$, then $M_{R,x}$ and $M_{R,y}$ are independent. We know that $H$ is $1-$Lipschitz, and thus, 
\[
|\partial_{x}H(M_{R})|^{4}\leq|J_{x}-J_{x}'|^{4} 
\]
This is bounded in expectation since we know that $\norm{J}_{4+\gd}<\infty$. Thus,  \eqref{lem:bigexpectation} is bounded above by 
\[
16C\eps(n,R)+C_{2}\mv1_{y\in B_{R}(x)}.
\]
Summing over $x,y$ completes the proof. 
\end{proof}
Lemma \ref{lem:CLTerrbd} was the last piece required to prove Theorem \ref{thm:CLT}.
\begin{proof}
We have already shown that
\[
\frac{1}{({\var{H(M_{n})}})^{3/4}}\left(\sum_{x\in \mathsf{V}_{n}\sqcup \mathsf{E}_{n}}\E|\partial_{x}H(M_{n})|^{3}\right)^{1/2}\to 0.
\]
The only step remaining is to specify the $R$ in Lemma \ref{lem:CLTerrbd}. The only requirement is that $R\to \infty$ as $n\to \infty$ and $R/n=o(1)$. A valid choice is yielded by $R=\lfloor \sqrt{n}\rfloor$. It is clear then that 
\[
\frac{1}{\sqrt{\var{H(M_{n})}}}\left(\sum_{x,y}c(x,y)\right)^{1/4}\leq C\eps(n,\sqrt{n})\to 0.
\]
\end{proof}
\subsection{Generalizations}\label{sec:generalization} 
We have defined the finite model on the torus and then taken a large $n$ limit so that translation invariance is guaranteed from the start. However, the methods described will work for any weight independent boundary condition (with suitable averaging over translates to yield translation invariance in limit). This method can also easily be adapted to the setting of amenable, unimodular, transitive graphs. This is because translation invariance and the Burton Keane argument are the key steps. Our methods should also be generalizable to obtain a `$1$ or $\infty$' type result for the cardinality of $\sf G(\mvJ)$  in the spirit of \cite{arguin_damron}, however we do not pursue this in this article.

Another direction in which our results can be generalized is related to the support of the distribution of the weights on the monomers. Our methods (directly) cover the important case of the weights on the edges being random variables with good distribution taking values in $[\gb,\infty)$ where $\gb<0$, but for the vertices the weights are identically 0. This is because in the perturbation step \Cref{eq:perturb_J}, we only perturb the edge weights, and hence absolute continuity for this case still holds. Another observation is that this case is equivalent to studying the minimum weight independent set on the line graphs of transitive, unimodular, amenable graphs, which is perhaps of independent interest. Our results directly imply uniqueness and a CLT for ground states of independent sets on such graphs.

\section{Open Questions}\label{sec:open}.

One obvious direction of generalizing \Cref{thm:main} is to extend the result for weights distribution given by any nonatomic measure. One particular case which is interesting is that of Uniform$(0,1)$ which our results do not cover. In particular, the step of making many edges inaccessible in the perturbation step breaks down in this setup.

Let $M$ be the infinite volume ground state for $\mvJ$.
A \textbf{critical droplet} for $x \in \sf \Sigma$ is the set $M\Delta M'$ where $M$ is the ground state for $\mvJ$ and $M'$ is the ground state $M_{x,\eps}$ where $\eps=1$ if $M$ is 0 at $x$ and vice-versa. In other words, we flip the status of $M$ at  $x$ and ask how the ground state is changed.
\begin{question}
Is the size of critical droplet at $\mv0$ finite almost surely?
\end{question}
The next question is about the size of $M\Delta M(p)$. Although \Cref{thm:perturbation} yields that $M \Delta M(p)$ has  finite components almost surely, the method is not quantitative. 

\begin{question}
    Let $L$ be the size of the component containing $\bf0$ of $M \Delta M(p)$. What can be said about the tail of $L$? 
\end{question}

\begin{question}
What can be deduced about the rate of convergence to $N(0,1)$ in the Central limit \Cref{thm:CLT}? This question is intimately tied to the rate of correlation decay of the ground state (\Cref{cor:decay_corr}). 
\end{question}

\section{Appendix}
\subsection{Mass Transport Principle}
In the study of translation invariant processes the {mass transport principle} (MTP) is a crucial tool. We state it below in the form used here, and refer the reader to \cite{LP:book,AL07} for more details.

Let us consider a random function ${\sf X}:{\sf \Sigma} \to \Gamma$ where $\Gamma$ is some complete separable metric space. Let us assume $\sf X$ is translation invariant in law. We will apply this for ${\sf X} = \Xi$ or ${\sf X} = X$ where $\Xi$  and $ X$ are as in \eqref{Xi}. A function $g : {\sf V} \times {\sf V} \times \Gamma^{{\sf \Sigma}} \to [0,\infty]$ is said to be invariant under diagonal action of $\dZ^d$ if for every $\alpha \in \dZ^d$, and every $\xi: {\sf \Sigma} \to \Gamma$,
$$
g(x,y,\xi) =  g(x+\alpha, y+\alpha, \xi+\alpha)
$$
where $\xi+\alpha:{\sf \Sigma} \to \Xi$ is defined as $(\xi+\alpha)_{i+\alpha} = \xi_{i}$ for all $i \in \Sigma$.
The function $g(x,y,\xi)$ can be thought of as the `mass sent from $x$ to $y$'. The mass transport principle states that in expectation (over the law of ${\sf X}$), the total mass out of the origin is the same as the total mass in.

\begin{lem}[Mass transport principle]\label{lem:MTP}
Let $g$ be a function as above which is invariant under the diagonal action of $\dZ^d$.  Then
\[
\dE\left(\sum_{y}g({\bf 0},y,{\sf X})\right)=\dE\left(\sum_{y}g(y,{\mv0},{\sf X})\right)
\]
\end{lem}
\subsection{Absolutely Continuous Perturbations}
In this section we finish the proof of \Cref{lem:abs_inf}. Let $\Sigma$ be a countable set. Recall the definition of $g_z$ from \Cref{def:good}, as well as $C(z,\ga)$.
\begin{lem}\label{lem:abs}
Consider $(\mvJ,U,U')\in [\gb,\infty)^{\Sigma}\times [0,1]^{2}$, where $(U,U')$ are independent uniform random variables, which are independent of $\mvJ$. Let $S\subset \Sigma$ be a fixed subset. Let $\mvZ\in [0,\infty)^\Sigma$ depend measurably on $(\mvJ,U,U')$ and assume that $\mvZ$ is identically 0 outside $S$, that is $Z_{x}=0$ for all $x\notin S$. Define $\tilde{\mvJ}$ by setting $\tilde{J}_{x}=J_{x}+Z_x$ for all $x\in \Sigma$. The distribution of $(\tilde{\mvJ},U,U')$ is absolutely continuous with respect to the distribution of $(\mvJ,U,U')$.  
\end{lem}
 We state the following measure theoretic fact before proving our lemma. Let $\mvz \in [0,\infty)^\Sigma$ be such that it is $0$ except at a finite fixed collection of indices $S \subset \gS$. The following equality is a standard computation for any Borel set $\cA \subset [\gb,\infty)^{\gS}$ and measurable $F$ such that the expectations are finite on both sides:
\begin{align}
\E \bigl(F(\mvJ+\mvz)\mv1_{\mvJ+\mvz\in \cA}\bigr)=\E\biggl(F(\mvJ)\mv1_{\mvJ \in \cA}\prod_{x\in \gS}g_{z_{x}}(J_{x}) \biggr)
\end{align}
Note that $g_{z_x} (J_x) =1$ if $x \not \in S$.

\begin{proof}[Proof of Lemma \ref{lem:abs}]
Let $\cA$ be a Borel subset of $[\gb,\infty)^{\gS}$, and $\cB$ be a Borel subset of $[0,1]^{2}$. For fixed $\mvz\in [0,\infty)^{S}$ and $\eps>0$, let $B_{\eps}(\mvz)$ denote the ball of radius $\eps$ centered at $\mvz$, appropriately truncated at the boundary. We know that 
\begin{align*}
\dP &\left((\tilde{\mvJ},U,U')\in \cA\times \cB,  \mvZ(\mvJ,U,U')\in B_{\eps}(\mvz)\right)\\ &=\dP\left((\mvJ,U,U')\in (\cA-\mv{Z})\times \cB,\mvZ(\mvJ,U,U')\in B_{\eps}(\mvz)\right)
\end{align*}

\begin{align*}
&=\E \biggl(\mv1_{(\mvJ,U,U')\in (\cA-\mvZ)\times \cB}\cdot \mv1_{\mvZ \in B_{\eps}(\mvz)}\biggr)\\
&=\E \biggl(\mv1_{(\mvJ,U,U') \in \cA\times \cB}\cdot \mv1_{\mvZ\in B_{\eps}(\mvz)}\cdot\prod_{x\in S}g_{Z_{x}}(J_{x})\biggr).
\end{align*}
By H\"older's inequality, we obtain an upper bound given by 
\begin{align*}
\max_{y\in B_{\eps}(z)}\E\biggl(\prod_{x\in S}g^{\ga}_{y_{i}}(J_{x})\biggr)^{1/\ga}\cdot & \dP ((\mvJ,U,U') \in \cA\times \cB)^{(\ga-1)/\ga} \\
& \leq \max_{\mvy\in B_{\eps}(\mvz)}C(y_{i},\ga)^{|S|/\ga}\cdot \dP((\mvJ,U,U')\in \cA\times \cB)^{\frac{\ga-1}{\ga}}.
\end{align*}
This is sufficient to conclude that if $\dP((\mvJ,U,U')\in\cA\times \cB)=0$, then $\dP((\tilde{\mvJ},U,U') \in \cA\times \cB)=0$ by a standard limiting argument involving $\mvz$. 
\end{proof}

\subsection{Acknowledgements} We would like to thank Partha Dey, Arnab Sen and Grigory Terlov for helpful discussions. KK acknowledges the Pacific Insititute for Mathematical Sciences, this work was done with the support of a PIMS Postdoctoral Fellowship.  \bibliography{monomer_dimer}

\begin{bibdiv}
\begin{biblist}

\bib{AW_90}{article}{
      author={Aizenman, Michael},
      author={Wehr, Jan},
       title={Rounding effects of quenched randomness on first-order phase transitions},
        date={1990},
     journal={Communications in mathematical physics},
      volume={130},
      number={3},
       pages={489\ndash 528},
}

\bib{ac1}{article}{
      author={Alberici, Diego},
      author={Contucci, Pierluigi},
       title={Solution of the monomer-dimer model on locally tree-like graphs. {R}igorous results},
        date={2014},
        ISSN={0010-3616,1432-0916},
     journal={Comm. Math. Phys.},
      volume={331},
      number={3},
       pages={975\ndash 1003},
         url={https://doi.org/10.1007/s00220-014-2080-3},
      review={\MR{3248055}},
}

\bib{ac2}{article}{
      author={Alberici, Diego},
      author={Contucci, Pierluigi},
      author={Mingione, Emanuele},
       title={A mean-field monomer-dimer model with randomness: exact solution and rigorous results},
        date={2015},
        ISSN={0022-4715,1572-9613},
     journal={J. Stat. Phys.},
      volume={160},
      number={6},
       pages={1721\ndash 1732},
         url={https://doi.org/10.1007/s10955-015-1306-x},
      review={\MR{3382766}},
}

\bib{AL07}{article}{
      author={Aldous, David},
      author={Lyons, Russell},
       title={Processes on unimodular random networks},
        date={2007},
        ISSN={1083-6489},
     journal={Electron. J. Probab.},
      volume={12},
       pages={no. 54, 1454\ndash 1508},
         url={http://dx.doi.org/10.1214/EJP.v12-463},
      review={\MR{2354165 (2008m:60012)}},
}

\bib{aldous_zeta_2}{article}{
      author={Aldous, David~J.},
       title={The {$\zeta(2)$} limit in the random assignment problem},
        date={2001},
        ISSN={1042-9832,1098-2418},
     journal={Random Structures Algorithms},
      volume={18},
      number={4},
       pages={381\ndash 418},
         url={https://doi.org/10.1002/rsa.1015},
      review={\MR{1839499}},
}

\bib{arguin_damron}{article}{
      author={Arguin, Louis-Pierre},
      author={Damron, Michael},
       title={On the number of ground states of the {E}dwards-{A}nderson spin glass model},
        date={2014},
        ISSN={0246-0203,1778-7017},
     journal={Ann. Inst. Henri Poincar\'{e} Probab. Stat.},
      volume={50},
      number={1},
       pages={28\ndash 62},
         url={https://doi.org/10.1214/12-AIHP499},
      review={\MR{3161521}},
}

\bib{ADNS_10}{article}{
      author={Arguin, Louis-Pierre},
      author={Damron, Michael},
      author={Newman, C.~M.},
      author={Stein, D.~L.},
       title={Uniqueness of ground states for short-range spin glasses in the half-plane},
        date={2010},
        ISSN={0010-3616,1432-0916},
     journal={Comm. Math. Phys.},
      volume={300},
      number={3},
       pages={641\ndash 657},
         url={https://doi.org/10.1007/s00220-010-1130-8},
      review={\MR{2736957}},
}

\bib{BM87}{article}{
      author={Bray, Alan~J},
      author={Moore, Michael~A},
       title={Chaotic nature of the spin-glass phase},
        date={1987},
     journal={Physical review letters},
      volume={58},
      number={1},
       pages={57},
}

\bib{BK}{article}{
      author={Burton, R.~M.},
      author={Keane, M.},
       title={Density and uniqueness in percolation},
        date={1989},
        ISSN={0010-3616,1432-0916},
     journal={Comm. Math. Phys.},
      volume={121},
      number={3},
       pages={501\ndash 505},
         url={http://projecteuclid.org/euclid.cmp/1104178143},
      review={\MR{990777}},
}

\bib{cao}{article}{
      author={Cao, Sky},
       title={Central limit theorems for combinatorial optimization problems on sparse {E}rdos-{R}\'{e}nyi graphs},
        date={2021},
        ISSN={1050-5164,2168-8737},
     journal={Ann. Appl. Probab.},
      volume={31},
      number={4},
       pages={1687\ndash 1723},
         url={https://doi.org/10.1214/20-aap1630},
      review={\MR{4312843}},
}

\bib{C_Normal_08}{article}{
      author={Chatterjee, Sourav},
       title={A new method of normal approximation},
        date={2008},
        ISSN={0091-1798,2168-894X},
     journal={Ann. Probab.},
      volume={36},
      number={4},
       pages={1584\ndash 1610},
         url={https://doi.org/10.1214/07-AOP370},
      review={\MR{2435859}},
}

\bib{chatterjee2023spin}{article}{
      author={Chatterjee, Sourav},
       title={Spin glass phase at zero temperature in the edwards-anderson model},
        date={2023},
     journal={arXiv preprint arXiv:2301.04112},
}

\bib{chatterjee_sen_mst}{article}{
      author={Chatterjee, Sourav},
      author={Sen, Sanchayan},
       title={Minimal spanning trees and {S}tein's method},
        date={2017},
        ISSN={1050-5164,2168-8737},
     journal={Ann. Appl. Probab.},
      volume={27},
      number={3},
       pages={1588\ndash 1645},
         url={https://doi.org/10.1214/16-AAP1239},
      review={\MR{3678480}},
}

\bib{DK_23}{article}{
      author={Dey, Partha~S.},
      author={Krishnan, Kesav},
       title={Disordered monomer-dimer model on cylinder graphs},
        date={2023},
        ISSN={0022-4715,1572-9613},
     journal={J. Stat. Phys.},
      volume={190},
      number={8},
       pages={Paper No. 146, 40},
         url={https://doi.org/10.1007/s10955-023-03159-7},
      review={\MR{4630528}},
}

\bib{dey2024random}{misc}{
      author={Dey, Partha~S.},
      author={Terlov, Grigory},
       title={Random optimization problems at fixed temperatures},
        date={2024},
}

\bib{Durbook}{book}{
      author={Durrett, Rick},
       title={Probability: theory and examples},
     edition={Fourth},
      series={Cambridge Series in Statistical and Probabilistic Mathematics},
   publisher={Cambridge University Press},
     address={Cambridge},
        date={2010},
        ISBN={978-0-521-76539-8},
      review={\MR{2722836 (2011e:60001)}},
}

\bib{SL_23}{article}{
      author={Lam, Wai-Kit},
      author={Sen, Arnab},
       title={Central limit theorem in disordered monomer-dimer model},
        date={2022},
     journal={arXiv preprint arXiv:2208.02151},
}

\bib{LP:book}{book}{
      author={Lyons, R.},
      author={Peres, Y.},
       title={Probability on trees and networks},
   publisher={Cambridge University Press},
        note={In preparation. Current version available at \hfil\break {\tt http://mypage.iu.edu/\string~rdlyons/}},
}

\bib{NS_01}{incollection}{
      author={Newman, C.~M.},
      author={Stein, D.~L.},
       title={Are there incongruent ground states in 2{D} {E}dwards-{A}nderson spin glasses?},
        date={2001},
      volume={224},
       pages={205\ndash 218},
         url={https://doi.org/10.1007/PL00005586},
        note={Dedicated to Joel L. Lebowitz},
      review={\MR{1868997}},
}

\bib{panchenko}{article}{
      author={Panchenko, Dmitry},
       title={The {P}arisi ultrametricity conjecture},
        date={2013},
        ISSN={0003-486X,1939-8980},
     journal={Ann. of Math. (2)},
      volume={177},
      number={1},
       pages={383\ndash 393},
         url={https://doi.org/10.4007/annals.2013.177.1.8},
      review={\MR{2999044}},
}

\bib{parisi}{article}{
      author={Parisi, Giorgio},
       title={Order parameter for spin-glasses},
        date={1983},
        ISSN={0031-9007},
     journal={Phys. Rev. Lett.},
      volume={50},
      number={24},
       pages={1946\ndash 1948},
         url={https://doi.org/10.1103/PhysRevLett.50.1946},
      review={\MR{702601}},
}

\bib{sherrington_kirkpatrick}{article}{
      author={Sherrington, David},
      author={Kirkpatrick, Scott},
       title={Solvable model of a spin-glass},
        date={1975Dec},
     journal={Phys. Rev. Lett.},
      volume={35},
       pages={1792\ndash 1796},
         url={https://link.aps.org/doi/10.1103/PhysRevLett.35.1792},
}

\bib{talagrand}{article}{
      author={Talagrand, Michel},
       title={The {P}arisi formula},
        date={2006},
        ISSN={0003-486X,1939-8980},
     journal={Ann. of Math. (2)},
      volume={163},
      number={1},
       pages={221\ndash 263},
         url={https://doi.org/10.4007/annals.2006.163.221},
      review={\MR{2195134}},
}

\end{biblist}
\end{bibdiv}
\bibliographystyle{abbrv}
\end{document}